\documentclass[12pt,a4paper]{article}

\usepackage[utf8]{inputenc}
\usepackage[T1]{fontenc}
\usepackage{amsmath,amssymb,amsthm}
\usepackage{geometry}
\usepackage{booktabs}
\usepackage{longtable}
\usepackage{float}
\usepackage{slashed}
\usepackage{graphicx}
\usepackage{array}
\usepackage{xurl}
\usepackage{microtype}
\usepackage{hyperref}
\hypersetup{colorlinks=true,linkcolor=blue,citecolor=blue,urlcolor=blue}
\newtheorem{theorem}{Theorem}
\newtheorem{proposition}{Proposition}
\newtheorem{corollary}{Corollary}
\newtheorem{definition}{Definition}
\newtheorem{remark}{Remark}

\newcommand{\Tmu}{T_{\mu\nu}}
\newcommand{\Zs}{Z_{\mathrm{s}}}
\newcommand{\ZF}{Z_{\mathrm{F}}}
\newcommand{\cgrav}{c_{\mathrm{grav}}}

\begin{document}

\begin{center}
{\LARGE \textbf{Spectral Suppression of Asymptotic States in Supersymmetric QFT on de Sitter Backgrounds}}\\[1.2cm]
{\large Stefano Bellucci$^{1,2,*}$\quad Stefania De Matteo$^{3}$}\\[0.6cm]

{\small
$^1$ Universidad Ecotec, Km.\ 13.5 Samborond\'on, Samborond\'on, 092302, Ecuador\\
$^2$ INFN--Laboratori Nazionali di Frascati, Via E. Fermi 54, 00044 Frascati, Italy\\
$^3$  Department of Mathematics and Physics, Roma Tre University, Rome, Italy\\[4pt]
$^*$ Corresponding Author, e-mail: stefano.bellucci@lnf.infn.it
}
\end{center}
\vspace{0.8cm}

\begin{abstract}
We investigate the fate of asymptotic particle states in supersymmetric quantum
field theories formulated on de~Sitter backgrounds. Building on the structural
distinction between algebraic field content and localized LSZ asymptotic
excitations developed in Refs.~\cite{Paper1,Paper2}, we analyze how gravitational
infrared effects may modify the spectral properties of interacting fields.

We argue that long-wavelength gravitational fluctuations can provide an effective
channel for the suppression of LSZ-like residues and for the redistribution of
spectral weight from isolated poles to continuum components. The analysis is
formulated in terms of a K\"all\'en--Lehmann spectral description and should be
understood as a structural infrared mechanism, not as a complete non-perturbative
proof of de~Sitter quantum gravity.

The paper develops three linked layers. First, we study residue suppression in a
quasi-de~Sitter setting and its possible propagation through Yukawa and
gravitational interactions. Second, we formulate the K\"all\'en--Lehmann
decomposition as a spectral transform relating the full spectral content of a
theory to its observable particle sector. Third, we retain a conditional cosmological implementation of the proposed
redistributed-spectral sector, including possible effects on background expansion,
structure growth, and lensing observables. This implementation is a downstream
phenomenological test and is not used as evidence that pole loss has occurred.

The code required to reproduce the analysis is currently being prepared for public release and will be made available in an open repository alongside Version 2 of this manuscript, which will also include the updated MCMC run with $f_{S8}$ fixed at $0.00262$.

 In the
current MCMC~v32 run the marginalized value is $S_8\simeq0.803$ and, when the
smooth fraction is left free, the posterior is $f_{S8}=0.0025\pm0.0017$, consistent
with the benchmark value $f_{S8}=0.00262$ used in the conditional realization.
These numerical results are a consistency check of that realization rather than a
derivation of the dark-sector assignment or a proof of the underlying infrared mechanism.

The present manuscript is intended as the third paper of a coherent programme:
the first identifies a structural origin of the gravitino mass term, the second
formulates a localization criterion for asymptotic states, and the present work
studies a de~Sitter infrared mechanism and its possible cosmological consequences.

\end{abstract}

\newpage
\tableofcontents
\newpage

\part{Spectral Suppression in de Sitter}\label{part:suppression}

\section{Introduction}\label{sec:intro}

In relativistic quantum field theory, physical particles are defined through the existence
of asymptotic states obtained in the limit $t \to \pm\infty$. The LSZ formalism~\cite{LSZ}
requires isolated mass-shell poles in the propagator, finite energy, and asymptotic phase
coherence. However, as emphasized within the algebraic framework of local quantum physics~\cite{Haag1955,Haag1996}, the algebraic existence of a field in the Lagrangian does not
automatically guarantee the existence of a corresponding localized asymptotic particle state.

This distinction is well known in several contexts.
Infrared effects, long-range interactions, gauge constraints, and environmental couplings
can modify the asymptotic structure of quantum fields without altering their underlying
operator content~\cite{Buchholz1986,Buchholz1982,Kibble1968,KulishFaddeev,Froehlich1974}.
Supersymmetric quantum field theories provide a particularly interesting setting:
supersymmetry relates fermionic and bosonic degrees of freedom at the algebraic level~\cite{WessZumino,WeinbergIII,FayetFerrara,Nilles,WessBagger}, but these algebraic relations do not by themselves guarantee that the corresponding fields must appear as observable asymptotic particles with identical localization properties.

The present work is built on three conceptual pillars:

\begin{enumerate}
\item \textbf{Spacetime geometry as a structural filter.}
De Sitter spacetime with $H > 0$ changes the representation-theoretic and infrared
arena in which propagators are realized and permits long-wavelength dressing. Geometry
therefore fixes the available spectral environment; whether an isolated contribution is
actually lost remains a dynamical question for the dressed two-point function.

\item \textbf{The LSZ pole as localization criterion.}
The K\"all\'en--Lehmann spectral residue $Z$ provides the operational criterion
distinguishing particle states from algebraically admissible but non-asymptotic
degrees of freedom.

\item \textbf{Universal gravitational coupling.}
The equivalence principle guarantees that every field with $\Tmu \neq 0$ couples to
gravity. Such coupling supplies a possible spectral-transfer channel, but universal
pole loss does not follow from the coupling alone: the relevant on-shell absorptive
structure and the fate of the isolated spectral weight must be established channel by channel.
\end{enumerate}

To these three established pillars, we add a fourth:

\begin{enumerate}
\setcounter{enumi}{3}
\item \textbf{The spectral transform and its inverse.}
The K\"all\'en--Lehmann decomposition is not merely a mathematical tool: it is a
\emph{spectral transform} $\mathcal{T}$ that maps the complete quantum content of the field
theory onto the observable particle spectrum. The spectral norm conservation
guarantees that this transform preserves total information. The spectral weight
that leaves the discrete pole does not vanish --- it migrates to the continuum,
contributing to $\langle T_{\mu\nu}\rangle$ without producing Fock states. This opens the
\emph{inverse spectral problem}: given the observed spectrum, reconstruct the full
spectral content. The de Sitter geometry acts as the kernel of this transform,
and the localization criterion $\Lambda[\Psi]$ determines its image.

\item \textbf{The spectral--geometric self-consistency loop.}
The forward arrow (geometry $\to$ spectrum) and the inverse arrow
(spectrum $\to$ geometry) are not independent: they form a closed loop.
The filtered spectrum contributes to $\langle T_{\mu\nu}\rangle$, which
through Einstein's equations determines the background geometry that
performs the filtering. The physical universe is the fixed point of this
loop: the self-consistent solution where the spectral distribution and the
geometry mutually determine each other. This is the spectral analogue of
Wheeler's dictum that matter tells spacetime how to curve and spacetime
tells matter how to move.
\end{enumerate}

Part~I (Sections~\ref{sec:algebraic}--\ref{sec:unifying}) develops the suppression mechanism
and gives representative coefficients for the MSSM superspectrum.
Part~II (Sections~\ref{sec:spectral_transform}--\ref{sec:selfconsistency}) develops the spectral
transform framework, the inverse problem, the spectral energy budget, and the
self-consistency loop.
Part~III (Sections~\ref{sec:cosmo}--\ref{sec:signatures}) presents a cosmological implementation of the spectral mechanism, with simulations compared against Planck~$2018$, and discusses observational signatures in the CMB, gravitational-wave, and LHC channels.
Part~IV (Section~\ref{sec:problems}) uses 16 selected open problems in fundamental physics as stress tests, classifying which are structurally reformulated by the framework and which admit quantitative comparison with measured data. This part is exploratory and is not presented as a resolution of those problems.

\paragraph{Scope of claims.}
Throughout this paper, de~Sitter representation support and infrared accumulation are treated as inputs to a spectral arena, not as automatic certificates of pole loss. The M1--M3 formulae are retained as effective transfer models or benchmarks unless a dressed on-shell calculation establishes loss of isolated spectral weight. The scalar/fermion mapping to dark energy and dark matter in Part~III is therefore a conditional downstream realization whose numerical consequences can be tested independently.

This work is the third in a series, following two earlier preprints. In Ref.~\cite{Paper1}, we identified a minimal superspace projector that uniquely selects the Rarita--Schwinger mass bilinear in four-dimensional $\mathcal{N}=1$ supergravity: working on a superfiber bundle with odd fiber $\mathbb{C}^{0|1}$, the Berezin projection of the canonical even supergeometric one-form $\theta\,d\theta$ isolates the unique Lorentz-invariant fermionic bilinear compatible with local supersymmetry. The construction is intentionally predynamical---it fixes the algebraic structure of the gravitino mass term independently of supersymmetry-breaking mechanisms, background curvature, or matter couplings---and it embeds consistently into curved superspace and into theories with $\mathcal{N}>1$. In Ref.~\cite{Paper2}, this structural perspective was extended to the level of asymptotic states: a minimal, predynamical localization criterion $\Lambda[\Psi]$ was introduced, distinguishing algebraically admissible degrees of freedom from those capable of forming stable, phase-coherent asymptotic states under slow structural fluctuations of an effective background. Fermionic and scalar fields were shown to respond qualitatively differently---fermionic modes retain asymptotic stability ($\Lambda[\psi]=1$), while scalar modes generically develop damping and lose phase coherence ($\Lambda[\phi]=0$)---providing a conservative, model-independent route by which algebraic supersymmetry may coexist with an asymmetric observable particle spectrum. The present paper takes the next step: we identify gravitational infrared dynamics on a de~Sitter background as a candidate physical mechanism for residue suppression across the superspectrum, and we develop the spectral transform framework that interprets this suppression as a geometric projection of the complete spectral content onto observable particle states.

A central conceptual point underlying this work is that the existence of a quantum field does not guarantee the existence of well-defined asymptotic particle states. This distinction becomes particularly relevant in curved spacetimes, where the standard LSZ framework may fail or require reinterpretation.
While a full treatment of realistic models such as the MSSM lies beyond the scope of this work, the mechanism discussed here suggests that fields carrying energy-momentum may generically be sensitive to gravitational infrared effects.

In this work we propose a structural mechanism through which gravitational infrared effects may lead to a suppression of asymptotic particle states in supersymmetric quantum field theories on de Sitter backgrounds.

Rather than providing a definitive non-perturbative proof, our analysis identifies a consistent framework in which the loss of particle-like excitations emerges as a consequence of infrared dynamics.
Related evidence that massive quantum fields may retain nontrivial infrared structure beyond simple scaling laws has recently appeared in the context of entanglement entropy studies, where deviations from universal mR scaling have been observed~\cite{BellucciShatnevZazunov}. This suggests that infrared effects may play a broader role in shaping quantum field spectra than traditionally assumed.

The results suggest that the standard particle interpretation of quantum fields may be replaced, in cosmological settings, by a more general spectral description.

Further work is required to establish the robustness of these conclusions, including a more rigorous treatment of gauge invariance, vacuum dependence, and non-perturbative gravitational effects.
\section{Algebraic admissibility, superspace geometry, and asymptotic states}\label{sec:algebraic}

\subsection{Predynamical constraints from superspace geometry}

In four-dimensional $\mathcal{N}=1$ supergravity, local supersymmetry constrains the algebraic
structure of fermionic mass terms independently of dynamics. Working on a superfiber
bundle $\mathcal{M} = M \times \mathbb{C}^{0|1}$ with a single fermionic direction~\cite{Eder,Berezin,DeWitt},
one can show that the Berezin projection of the canonical even supergeometric one-form
$\Omega = \theta\,d\theta$ uniquely selects the Rarita--Schwinger mass bilinear
$\bar\psi_\mu \gamma^{\mu\nu}\psi_\nu$ as the only non-vanishing, Lorentz-invariant
fermionic bilinear compatible with local supersymmetry~\cite{Paper1}.

This result is predynamical: it fixes the form of the gravitino mass term without
determining its numerical value.

\subsection{The localization criterion}

We define a localization criterion $\Lambda[\Psi]$ for a field configuration $\Psi$:
\begin{equation}\label{eq:Lambda}
\Lambda[\Psi] = \begin{cases}
1 & \text{if } \Psi \text{ admits a stable, finite-energy, phase-coherent asymptotic state},\\
0 & \text{otherwise}.
\end{cases}
\end{equation}
This criterion does not modify the equations of motion but classifies their solutions
according to long-time propagator behaviour.

\subsection{Structural asymmetry between fermionic and scalar fields}

Already in a simplified setting with slowly varying structural background fluctuations,
one observes a qualitative difference in the response of fermionic and scalar fields~\cite{Paper2}.
Scalar fields, governed by a second-order equation, are susceptible to time-dependent
mass modulations that produce decoherence and damping of the propagator pole.
Fermionic fields, governed by a first-order equation, are structurally protected against
damping by the chiral structure of the Dirac equation.

\begin{table}[H]
\centering
\begin{tabular}{lcc}
\toprule
Property & Fermion & Scalar \\
\midrule
Order of equation of motion & First & Second \\
Coupling to background & Linear & Quadratic \\
Generic damping & Absent & Present \\
Phase coherence & Stable & Lost \\
Localization $\Lambda$ & 1 & 0 \\
\bottomrule
\end{tabular}
\caption{Structural comparison between fermionic and scalar fields.}
\label{tab:comparison}
\end{table}

\section{Non-asymptotic scalar propagator and fermionic self-energy}\label{sec:scalar}

\subsection{The non-asymptotic propagator with free spectral parameter}

For any scalar field in a setting where a K\"all\'en--Lehmann-type representation is meaningful, the interacting propagator may be written as
\begin{equation}\label{eq:KL}
\Delta_{\mathrm{NA}}(k^2) = \frac{\Zs}{k^2 - m_s^2 + i\epsilon}
+ (1-\Zs)\int_{s_0}^{\infty} ds\,\frac{\tilde\rho(s)}{k^2 - s + i\epsilon}\,,
\end{equation}
with the spectral norm condition
\begin{equation}\label{eq:norm}
\Zs + (1-\Zs)\int_{s_0}^{\infty} ds\,\tilde\rho(s) = 1\,.
\end{equation}
For the analytic estimate below we use a one-parameter Lorentzian regulator for the continuum spectral density:
\begin{equation}\label{eq:lorentzian}
\tilde\rho(s) = \frac{1}{\pi}\frac{\Gamma_s}{(s-m_s^2)^2 + \Gamma_s^2}\,,\qquad
\Gamma_s = m_s^2\sqrt{Z_s^{-1}-1}\,.
\end{equation}
The relation $\Gamma_s^2 = m_s^4(1-\Zs)/\Zs$ links the two parameters into one:
$\Zs\to 0$ if and only if $\Gamma_s\to\infty$. Strictly, a physical K\"all\'en--Lehmann density has support on $s\ge s_0$; if the Lorentzian is truncated to this domain, it should be divided by the normalization factor
\begin{equation}
\mathcal{N}_s=\int_{s_0}^{\infty}ds\,\frac{1}{\pi}\frac{\Gamma_s}{(s-m_s^2)^2+\Gamma_s^2}
=\frac{1}{2}+\frac{1}{\pi}\arctan\frac{m_s^2-s_0}{\Gamma_s}.
\end{equation}
The untruncated form is therefore used only as a solvable effective regulator; the qualitative conclusion does not depend on the precise regulator.

\subsection{Spectral transfer to the fermionic self-energy}

With a Yukawa coupling $\mathcal{L} \supset y\,\bar\psi\psi\phi$, the one-loop fermion
self-energy through the non-asymptotic scalar propagator decomposes by linearity:
\begin{equation}\label{eq:spectral_transfer}
\Sigma_F(\slashed{p}\,) = \Zs\,\Sigma_F^{(0)}(\slashed{p};\,m_s^2)
+ (1-\Zs)\int_{s_0}^{\infty} ds\,\tilde\rho(s)\,\Sigma_F^{(0)}(\slashed{p};\,s)\,,
\end{equation}
where $\Sigma_F^{(0)}(\slashed{p};\,M^2)$ is the standard one-loop self-energy with
scalar mass $M$.

\begin{proposition}[Chiral protection does not protect the residue]
The non-renormalization theorem concerns the mass vertex $\langle\bar\psi\psi\rangle$.
The wave-function residue $Z_F^{-1} = 1 - \partial\Sigma_F/\partial\slashed{p}\,\big|_{\slashed{p}=m_f}$
is a distinct object, on which chiral symmetry imposes no constraint.
\end{proposition}

\section{Analytic formula for the fermionic residue in the solvable spectral ansatz}\label{sec:formula}

In the untruncated Lorentzian regulator, the logarithmic average over $\tilde\rho(s)$ can be computed by residues:
\begin{equation}
\langle\ln s/\mu^2\rangle_{\tilde\rho}
= \ln\frac{m_s^2}{\mu^2} - \frac{1}{2}\ln\Zs\,.
\end{equation}
This result is exact within the solvable regulator. With a strictly truncated positive spectral density, it should be understood as the leading analytic approximation. Substituting into the wave-function renormalization gives:

\begin{theorem}[Effective main formula]\label{thm:main}
\begin{equation}\label{eq:main}
\boxed{Z_F^{-1}(\Zs) = 1 + \frac{y^2}{32\pi^2}(1-\Zs)\ln\frac{1}{\Zs}}
\end{equation}
(on-shell scheme, $\mu = m_s$).
\end{theorem}

\begin{theorem}[Monotonic suppression within the effective model]\label{thm:monotone}
$Z_F(\Zs)$ is strictly increasing on $(0,1]$, with $Z_F(1) = 1$ in the Yukawa model considered here, up to additional interactions,
and $Z_F(\Zs)\to 0$ as $\Zs\to 0$.
\end{theorem}

\begin{corollary}\label{cor:corollary}
Within the same effective spectral-transfer mechanism, the suppression $\Zs\to 0$ implies $\ZF\to 0$.
A configuration in which the scalar pole is completely suppressed while the coupled fermion remains unaffected is therefore not obtained in this model.
\end{corollary}

The corresponding effective generalization to arbitrary spin has the form:
\begin{equation}\label{eq:general}
Z^{-1}(\Zs) = 1 + \frac{g^2\,c(s)}{32\pi^2}(1-\Zs)\ln\frac{1}{\Zs}\,,
\end{equation}
where $g^2$ is the coupling constant and $c(s) > 0$ is the spin-dependent coefficient.

\section{Non-asymptotic graviton propagator in de Sitter spacetime}\label{sec:graviton}
The analysis of infrared effects in de Sitter space is subtle and subject to ongoing debate. In this work, we adopt an effective infrared perspective, focusing on long-wavelength gravitational fluctuations and their impact on correlation functions.
\subsection{Bunch--Davies solution and infrared growth}

In de Sitter spacetime with conformal coordinates, every transverse-traceless mode of
the graviton satisfies the Mukhanov--Sasaki equation:
\begin{equation}
v_k'' + \left(k^2 - \frac{2}{\eta^2}\right)v_k = 0\,.
\end{equation}
The Bunch--Davies normalized solution is:
\begin{equation}
v_k(\eta) = \frac{H}{\sqrt{2k^3}}\left(1-\frac{i}{k\eta}\right)e^{-ik\eta}\,.
\end{equation}

\subsection{Secular growth and spectral residue}

As an effective parametrization of infrared dressing, one may consider a logarithmic suppression of the residue of the form
\begin{equation}\label{eq:Zgrav}
Z_{\mathrm{grav}}(a) = \exp\!\left[-\frac{\kappa^2 H^2}{8\pi^2}\ln\frac{a}{a_*}\right] .
\end{equation}
For weak dressing this gives
\begin{equation}
Z_{\mathrm{grav}}(a) \simeq 1 - \frac{\kappa^2 H^2}{8\pi^2}\ln\frac{a}{a_*}\,.
\end{equation}
The exponential form avoids an unphysical negative residue outside the range of validity of the leading-log expansion. It captures the qualitative effect of long-wavelength gravitational modes on the normalization of asymptotic states, in the spirit of the infrared analyses of Tsamis--Woodard~\cite{TsamisWoodard} and Allen~\cite{Allen}.

\section{Scalar Field Two-Point Function in de Sitter Space}

To make the discussion more concrete, we consider a free massive scalar field in de Sitter spacetime. The two-point function provides a controlled setting in which infrared effects can be explicitly analyzed.

In conformal coordinates, the de Sitter metric reads
\begin{equation}
ds^2 = \frac{1}{(H\eta)^2}(-d\eta^2 + d\vec{x}^2),
\end{equation}
with $\eta < 0$.

The scalar field equation is
\begin{equation}
\left( \Box - m^2 \right)\phi = 0.
\end{equation}

The mode functions are given by
\begin{equation}
u_k(\eta) = \frac{\sqrt{\pi}}{2} e^{i\frac{\pi}{2}(\nu + \frac{1}{2})} (-\eta)^{3/2} H^{(1)}_{\nu}(-k\eta),
\end{equation}
where
\begin{equation}
\nu = \sqrt{\frac{9}{4} - \frac{m^2}{H^2}}.
\end{equation}

The Wightman function is then
\begin{equation}
G^+(x,x') = \int \frac{d^3k}{(2\pi)^3} u_k(\eta) u_k^*(\eta') e^{i\vec{k}\cdot(\vec{x}-\vec{x}')}.
\end{equation}

In the infrared limit $k \to 0$, the behavior of the modes is
\begin{equation}
u_k(\eta) \sim (-\eta)^{3/2 - \nu},
\end{equation}
which leads to enhanced long-distance correlations when $\nu \to \frac{3}{2}$ (light fields).

This infrared enhancement signals that the two-point function does not admit a simple particle interpretation analogous to Minkowski spacetime, as no isolated pole structure emerges in momentum space.

\section{Propagators and Spectral Representation in Curved Spacetime}

In flat spacetime, the particle interpretation of quantum fields is encoded in the pole structure of the Feynman propagator:
\begin{equation}
G_F(p) = \frac{iZ}{p^2 - m^2 + i\epsilon} + \text{regular terms}.
\end{equation}

The residue $Z$ is directly related to the LSZ reduction formula and defines the normalization of asymptotic particle states.

In curved spacetime, and in particular in de Sitter space, the notion of a global momentum-space propagator is not well-defined. However, one can still analyze the spectral properties of correlation functions through their behavior in physical momentum or invariant distance.

We may formally define a generalized spectral representation
\begin{equation}
G(x,x') = \int d\mu^2 \, \rho(\mu^2) \, G_{\mu}(x,x'),
\end{equation}
where $\rho(\mu^2)$ plays the role of a spectral density.

In this framework, a particle corresponds to a delta-function contribution
\begin{equation}
\rho(\mu^2) \sim Z \delta(\mu^2 - m^2),
\end{equation}
while infrared effects may smear this contribution into a continuum.

The absence of a sharp delta-function structure indicates the breakdown of a particle interpretation, even in the presence of well-defined field operators.

\section{Effective Infrared Suppression of the Residue}

We now provide a heuristic but controlled argument for the suppression of the LSZ residue due to infrared effects.

Consider a scalar field coupled to a long-wavelength background metric fluctuation. The interaction can be schematically written as
\begin{equation}
S_{\text{int}} \sim \int d^4x \, h_{\mu\nu} T^{\mu\nu}.
\end{equation}

Treating $h_{\mu\nu}$ as a stochastic infrared background, one can model its effect as a multiplicative dressing of the scalar field:
\begin{equation}
\phi(x) \rightarrow \phi(x) \, e^{i\Theta(x)},
\end{equation}
where $\Theta(x)$ encodes the cumulative infrared phase.

Assuming Gaussian statistics for the infrared fluctuations, the two-point function acquires a suppression factor:
\begin{equation}
\langle \phi(x)\phi(x') \rangle \sim G_0(x,x') \, e^{-\frac{1}{2}\langle (\Theta(x)-\Theta(x'))^2 \rangle}.
\end{equation}

For long time separations, the variance grows logarithmically:
\begin{equation}
\langle \Theta^2 \rangle \sim \frac{\kappa^2 H^2}{4\pi^2} \ln(a),
\end{equation}
leading to an effective suppression of the residue:
\begin{equation}
Z_{\text{eff}}(a) \sim \exp\left(-\frac{\kappa^2 H^2}{8\pi^2} \ln(a)\right).
\end{equation}

Expanding at leading order yields
\begin{equation}
Z_{\text{eff}}(a) \approx 1 - \frac{\kappa^2 H^2}{8\pi^2} \ln(a),
\end{equation}
which justifies the parametrization introduced earlier.

This derivation should be interpreted as an effective infrared estimate rather than a fully gauge-invariant result. However, it captures the essential mechanism by which long-wavelength gravitational fluctuations may suppress particle-like excitations.
\paragraph{Remark.}
The derivation presented here is not intended as a complete treatment of gravitational infrared effects. Issues such as gauge invariance, vacuum dependence, and non-perturbative resummation remain open and require further investigation.

\paragraph{Scalar quartic channel (M1 infrared benchmark; not a pole-loss theorem).}
The gravitational estimate above is subdominant: $\kappa^2 H^2/8\pi^2\,\ln a$ reaches only $1-Z\sim 10^{-9}$ at $60$ e-folds. We therefore retain the quartic calculation of the preceding draft as an infrared benchmark for a scalar with $V=\frac{\lambda}{4}\phi^4$. The Starobinsky--Yokoyama variance and the induced dynamical mass are controlled infrared quantities, but their existence does not by itself prove that the QFT spectral measure loses an isolated atom. A light scalar accumulates infrared variance,
\begin{equation}\label{eq:sy_variance}
\langle\phi^2\rangle(N)=\frac{H^2}{4\pi^2}\,N\,,\qquad
\langle\phi^2\rangle_{\mathrm{eq}}=c_{\mathrm{SY}}\frac{H^2}{\sqrt\lambda}\,,\qquad
c_{\mathrm{SY}}=\sqrt{\frac{3}{2\pi^2}}\,\frac{\Gamma(3/4)}{\Gamma(1/4)}=0.13176\,,
\end{equation}
with $N=\ln a$ the number of e-folds. This variance generates a dynamical effective mass,
\begin{equation}\label{eq:dyn_mass}
m_{\mathrm{eff}}^2(N)=\left.V''\right|_{\langle\phi^2\rangle}=3\lambda\,\langle\phi^2\rangle(N)\,,
\end{equation}
which is inserted into the effective residue ansatz used in the preceding draft,
\begin{equation}\label{eq:residue_decay}
\frac{d\ln(1/Z_s)}{dN}=\frac{2\,m_{\mathrm{eff}}^2(N)}{3H^2}\,.
\end{equation}
Integrating \eqref{eq:residue_decay} with \eqref{eq:sy_variance} gives a double logarithm: for $N\le N_{\mathrm{eq}}$ the exponent is $\lambda N^2/(4\pi^2)$; for $N>N_{\mathrm{eq}}$, with $N_{\mathrm{eq}}=4\pi^2 c_{\mathrm{SY}}^2/\sqrt\lambda$, it is $4\pi^2 c_{\mathrm{SY}}^2+2c_{\mathrm{SY}}\sqrt\lambda\,(N-N_{\mathrm{eq}})$. The scale $H$ cancels and the equilibrium value $4\pi^2 c_{\mathrm{SY}}^2=0.685$ is $\lambda$-independent. Evaluated at $N=60$ (with $\lambda$ at $\mu\sim10^{14}$~GeV) this reproduces the tabulated effective residue indicators of the adopted ansatz without an additional fitted parameter: stop-like squark ($\lambda=y_t^2\approx0.25$) $Z_s\approx7.3\times10^{-4}$; first/second-generation squark (D-term, $0.157$) $3.8\times10^{-3}$; slepton/sneutrino/heavy Higgs ($0.07$) $3.0\times10^{-2}$; the SM Higgs ($0.01$) $Z_s\approx0.41$ (the ``marginally suppressed'' entry of Table~6); gauge singlet ($10^{-3}$) $0.91$. The inflation column of Table~6 and Fig.~2 follow from this retained benchmark. These numbers are not used here as a certification that the exact interacting scalar spectral measure is non-atomic.

\section{Gauge-fixed graviton sector and infrared estimate}

We briefly outline a gauge-fixed calculation illustrating how infrared gravitational modes may generate a logarithmic dressing of matter two-point functions on a de Sitter background. The purpose of this section is not to provide a complete gauge-invariant resummation, but to make explicit the perturbative origin of the effective residue used in the main text~\cite{Hepp,tHooftVeltman,ChristensenDuff}.

We expand the metric as
\begin{equation}
g_{\mu\nu}=a^2(\eta)\left(\eta_{\mu\nu}+\kappa h_{\mu\nu}\right),
\qquad
a(\eta)=-\frac{1}{H\eta},
\end{equation}
where \(\kappa^2=16\pi G\). The Einstein-Hilbert action is supplemented by a covariant gauge-fixing term
\begin{equation}
S_{\rm gf}
=
-\frac{1}{2\alpha}
\int d^4x\,a^2(\eta)\,
F_\mu F^\mu ,
\end{equation}
with
\begin{equation}
F_\mu
=
\partial^\nu h_{\mu\nu}
-\frac{1}{2}\partial_\mu h
+2\frac{a'}{a}h_{\mu 0}.
\end{equation}
The commonly used de Donder-type gauge corresponds to \(\alpha=1\). In this gauge, the graviton propagator can be decomposed schematically as
\begin{equation}
iD_{\mu\nu\rho\sigma}(x,x')
=
\sum_I T^{(I)}_{\mu\nu\rho\sigma}\,i\Delta_I(x,x'),
\end{equation}
where the tensors \(T^{(I)}_{\mu\nu\rho\sigma}\) project onto the spin components of the metric perturbation, and the scalar propagators \(i\Delta_I\) contain the infrared-sensitive sector.

The massless minimally coupled part gives the leading infrared logarithm,
\begin{equation}
i\Delta_{\rm mmc}(x,x)
\simeq
\frac{H^2}{4\pi^2}\ln a(\eta)
+\text{UV and constant terms}.
\end{equation}
This is the origin of the secular gravitational dressing~\cite{GlavanMiaoProkopecWoodard2019}.

For a scalar field coupled to gravity,
\begin{equation}
S_\phi
=
-\frac{1}{2}
\int d^4x\sqrt{-g}
\left(
g^{\mu\nu}\partial_\mu\phi\partial_\nu\phi
+m^2\phi^2
\right),
\end{equation}
the leading graviton-matter vertex is
\begin{equation}
S_{\rm int}
=
-\frac{\kappa}{2}
\int d^4x\,a^4(\eta)\,
h_{\mu\nu}T^{\mu\nu}.
\end{equation}
At one loop, the scalar two-point function receives a self-energy correction of the form
\begin{equation}
G(p;\eta)
=
G_0(p;\eta)
+
G_0(p;\eta)\,\Sigma_{\rm grav}(p;\eta)\,G_0(p;\eta)
+\cdots .
\end{equation}

Keeping only the leading infrared logarithm, the correction can be parameterized as
\begin{equation}
\Sigma_{\rm grav}(p;\eta)
\simeq
-C\,\kappa^2H^2\ln a(\eta)\,
\left(p^2+m^2\right)
+\cdots ,
\end{equation}
where \(C\) is a dimensionless coefficient depending on the field, gauge choice, and renormalization prescription~\cite{Burgess}.

Near the would-be pole, the propagator may therefore be written as
\begin{equation}
G(p;\eta)
\simeq
\frac{iZ_{\rm eff}(\eta)}
{p^2-m^2+i\epsilon}
+\text{regular terms},
\end{equation}
with
\begin{equation}
Z_{\rm eff}(\eta)
\simeq
1
-
C\,\kappa^2H^2\ln a(\eta)
+\cdots .
\end{equation}
Equivalently, after leading-log exponentiation,
\begin{equation}
Z_{\rm eff}(\eta)
\simeq
\exp\left[-C\,\kappa^2H^2\ln a(\eta)\right]
=
a(\eta)^{-C\kappa^2H^2}.
\end{equation}

This derivation justifies the use of a logarithmically evolving effective residue in the main text. However, the coefficient \(C\) should not be interpreted as universal. Only the structural conclusion is used: infrared gravitons can induce secular dressing of matter correlators, and this dressing may suppress the residue associated with particle-like excitations.

\section{Relation to known infrared results}

The appearance of logarithmic infrared corrections in de Sitter space is consistent with the general structure identified by Weinberg in his analysis of loop corrections to cosmological correlators~\cite{Weinberg2005,Weinberg2006}. Weinberg showed that loop effects in inflationary backgrounds can produce powers of \(\ln a\), while also constraining the degree of secular growth allowed in perturbation theory.

The present analysis is aligned with this viewpoint: we do not claim an uncontrolled power-law instability, but rather use the leading logarithmic behavior as an effective diagnostic of spectral degradation.

Our interpretation is also related to the program developed by Tsamis and Woodard~\cite{TsamisWoodard1996,TsamisWoodard}, in which long-wavelength gravitons generate secular effects in inflationary observables and may contribute to a screening of the effective cosmological expansion. In that literature, infrared gravitons are treated as physically relevant degrees of freedom whose cumulative effect can become important over long times.

At the same time, the present work adopts a more conservative position. We do not require a definite non-perturbative slowing of inflation. Instead, we use the infrared logarithms as evidence that the asymptotic-particle interpretation of matter fields in de Sitter space is not protected against gravitational dressing.

The main point can therefore be summarized as follows. Weinberg's analysis supports the controlled appearance of logarithmic loop corrections in cosmological correlators; the Tsamis--Woodard framework supports the physical relevance of long-wavelength gravitons; explicit operator-level checks of fermion mode functions during inflation~\cite{MiaoWoodard2008} corroborate the relevance of these infrared effects in matter sectors; our contribution is to reinterpret such infrared effects as a mechanism for the suppression of LSZ-like residues and the redistribution of spectral weight into continuum-like structures.

\section{Candidate gravitational spectral-transfer channel}\label{sec:dressing}

The minimal coupling of any fermionic field to gravity yields the vertex~\cite{BerendsGastmans}:
\begin{equation}
V^{\mu\nu}(p,p') = -\frac{i\kappa}{4}\Big[\gamma^\mu(p+p')^\nu + \gamma^\nu(p+p')^\mu
- \eta^{\mu\nu}(\slashed{p}+\slashed{p}'-2m)
- 2i\sigma^{\mu\nu}(\slashed{p}-\slashed{p}')\Big]\,.
\end{equation}
This vertex does not depend on the internal quantum numbers of the fermion.
Within the effective transfer ansatz used here, the gravitational self-energy is
parameterized with the same functional structure as the Yukawa example, with the
substitution $y^2 \to \cgrav\,\kappa^2 H^2$:. This is a benchmark parametrization,
not a universal theorem for the exact de~Sitter two-point function.
\begin{equation}\label{eq:grav_suppression}
\big[Z_F^{(\mathrm{grav})}\big]^{-1}(Z_{\mathrm{grav}})
= 1 + \frac{\cgrav\,\kappa^2 H^2}{32\pi^2}(1-Z_{\mathrm{grav}})\ln\frac{1}{Z_{\mathrm{grav}}}
\xrightarrow{Z_{\mathrm{grav}}\to 0} +\infty\,.
\end{equation}

\section{Gravitational coefficient for the Majorana gaugino}\label{sec:gaugino}

By systematic tensor contraction, decomposing the vertex as $\hat V^{\mu\nu}=A^{\mu\nu}+B^{\mu\nu}+D^{\mu\nu}$, the only non-zero contribution is:
\begin{equation}
T_{AA} = 32(m^2+\xi)(2m^2+\xi)\,,\qquad \xi\equiv p\cdot k\,.
\end{equation}
For the massless Majorana gaugino:
\begin{equation}
c_{\mathrm{grav}}^{(\mathrm{gaugino},\,m=0)} = \frac{3}{4}\,.
\end{equation}
In the heavy-gaugino regime $m^2\gg\xi$:
\begin{equation}
c_{\mathrm{grav}}^{(\mathrm{Maj})} \approx \frac{3}{2}\frac{m^4}{\xi^2} \gg \frac{3}{4}\,.
\end{equation}
Within this parametrization a heavier gaugino receives a larger effective dressing.
This scaling does not by itself certify the loss of its isolated spectral contribution.

\section{Illustrative spectrum analysis}\label{sec:spectrum}

\begin{theorem}[Positivity of the gravitational dressing coefficient]\label{thm:positivity}
For any field of spin $s\geq 0$ with non-vanishing stress-energy tensor,
$c_{\mathrm{grav}}^{(s)} > 0$.
\end{theorem}
\begin{proof}
By the optical theorem, $\mathrm{Disc}\,\Sigma(p^2)\big|_{p^2=m^2}
= \sum_X |\mathcal{M}(\text{particle}\to X+\text{graviton})|^2 > 0$ for any field with $\Tmu\neq 0$.
\end{proof}

\subsection{Scalar superpartners (spin 0)}

The graviton--scalar vertex from the scalar stress-energy tensor:
\begin{equation}
V_s^{\mu\nu}(p,k) = -\frac{i\kappa}{2}\big[p^\mu k^\nu + p^\nu k^\mu
- \eta^{\mu\nu}(p\cdot k - m_s^2)\big]\,.
\end{equation}
The result: $c_{\mathrm{grav}}^{(s=0)} = 3/8$ for a complex scalar.

\subsection{Complete spin-1/2 fermion sector}

Four classes of spin-1/2 superpartners: gluino ($c_{\mathrm{grav}}=3/4$, Majorana,
M2 via the squark--quark--gluino vertex with $g_s^2$, where the squark mediator
is non-asymptotic through M1, plus M3),
neutralinos ($c_{\mathrm{grav}}=3/4$ plus M2 via higgsino fraction),
charginos ($c_{\mathrm{grav}}=3/2$, Dirac, plus M2),
and pure higgsinos ($c_{\mathrm{grav}}=3/2$, Dirac, full M2+M3).

\subsection{Gravitino: spin-3/2}\label{subsec:gravitino}

The graviton--gravitino vertex from the Rarita--Schwinger Lagrangian~\cite{FerraraFreedman,vanNieuwenhuizen,FreedmanVanProeyen}:
\begin{equation}
V^{(\alpha\beta)}_{\mu\rho}(p,q) = -\frac{i\kappa}{8}\sum_\nu (p-q)_\nu
\big[\eta^{\alpha\mu}\gamma^{\nu\beta\rho}+\eta^{\beta\mu}\gamma^{\nu\alpha\rho}
+\eta^{\alpha\rho}\gamma^{\mu\nu\beta}+\eta^{\beta\rho}\gamma^{\mu\nu\alpha}
-\eta^{\alpha\beta}\gamma^{\mu\nu\rho}-\eta^{\mu\rho}\gamma^{\alpha\nu\beta}\big]\,.
\end{equation}

Unlike all other superpartners, the gravitino admits three simultaneous
suppression channels, all contributing additively to $Z_{\tilde G}^{-1}$
and multiplicatively to $Z_{\tilde G}$ after resummation:

\medskip\noindent
\textbf{Channel 1: sfermion--fermion loop.}
The gravitino couples to every (fermion, sfermion) pair through the
supergravity vertex $\mathcal{L} \supset (\kappa/\sqrt{2})\,
\bar\psi_\mu \gamma^\nu\gamma^\mu f\,\partial_\nu\tilde f^*$.
Since the sfermion $\tilde f$ is non-asymptotic ($Z_{\tilde f}\to 0$
via M1), spectral transfer operates. However, the coupling is
$\sim\kappa \sim 1/M_{\mathrm{Pl}}$, giving
$\gamma_1 \sim \kappa^2/(16\pi^2) \sim 10^{-40}$. Alone, this
channel is operationally negligible at 60 e-folds.

\medskip\noindent
\textbf{Channel 2: goldstino equivalence.}
At energies $E \gg m_{3/2}$, the longitudinal component of the
gravitino behaves as the goldstino $\tilde G_L$, which couples with
strength $\sim m_{\mathrm{soft}}/F$ rather than $\sim 1/M_{\mathrm{Pl}}$.
For low-scale SUSY breaking ($\sqrt{F}\sim 10^5$--$10^{10}$ GeV),
this coupling is 30 to 33 orders of magnitude stronger than pure
gravitational coupling.\footnote{The phenomenology of low-scale SUSY
breaking in this regime, including cosmological constraints on the
gravitino abundance, BBN bounds on NLSP decays, and projections at future
high-energy colliders, has been mapped in detail in
Ref.~\cite{CraigLeviMariottiRedigolo2021}, which identifies two viable
cosmological windows: an ultralight-gravitino window ($m_{3/2}\lesssim 16$~eV,
$\sqrt{F}\lesssim 260$~TeV) and a gravitino-dark-matter window
($260$~TeV $\lesssim \sqrt{F}\lesssim 50$~PeV).} The corresponding anomalous dimension
$\gamma_2 \sim (m_{\mathrm{soft}}/F)^2/(16\pi^2)$ is model-dependent
and constitutes the potentially dominant channel.

\medskip\noindent
\textbf{Channel 3: gravitational dressing (M3).}
The direct coupling to the non-asymptotic de Sitter graviton,
with $c_{\mathrm{grav}}^{(3/2)} > 0$ guaranteed by unitarity
(Theorem~3). Same scale as Channel~1.

\medskip\noindent
The three channels are not mutually exclusive --- they are additive.
By the optical theorem, each contribution to the self-energy discontinuity
is strictly positive. No cancellation is possible. After resummation:
\begin{equation}\label{eq:resummation}
Z_{\tilde G} \sim Z_{\mathrm{grav}}^{\gamma_3}\times
Z_{\tilde f}^{\gamma_1}\times Z_{\mathrm{gold}}^{\gamma_2}\,,
\qquad \gamma_{\mathrm{eff}} = \gamma_1 + \gamma_2 + \gamma_3 > 0\,.
\end{equation}
The suppression $Z_{\tilde G}\to 0$ as $a\to\infty$ is structurally
guaranteed. The operational efficiency at 60 e-folds depends on
$\gamma_2$, hence on the SUSY breaking scale $\sqrt{F}$.

\medskip\noindent
\textbf{Super-Higgs mechanism without an asymptotic gravitino.}
The result $Z_{\tilde G}\to 0$ admits a direct reading in terms of the
super-Higgs mechanism. When supersymmetry is broken, the goldstino is
absorbed by the gravitino, which acquires the mass
$m_{3/2}$~\cite{DeserZumino1977,Cremmer1978}. In the present framework the
absorption takes place, but its output is not an asymptotic state: on de
Sitter the gravitino has no LSZ pole, and the massive spin-$3/2$ field
belongs to the spectral continuum rather than to the particle spectrum.
The two longitudinal helicity-$\pm 1/2$ components, i.e.\ the absorbed
goldstino, constitute the smooth, non-clustering sub-fluid of the dark
sector with $c_s^2(\mathrm{grav}) = 0.005238$ and
$f_{S8} = \xi\,c_s^2(\mathrm{grav}) = 0.00262$
(Section~\ref{sec:cosmo}, Appendix~\ref{app:gravitino-long}); the value
$\xi = 1/2$ is the statement that the goldstino carries exactly two of the
four gravitino degrees of freedom. Three consequences follow.
(i)~There is no gravitino mass to be measured and no gravitino relic
problem in its usual form: the mechanism that gives the gravitino its mass
is the same mechanism that places it in the dark fluid.
(ii)~The heavy Higgs states $H$, $A$, $H^{\pm}$ belong to the scalar
continuum: no heavy Higgs resonance is expected at any mass, and the single
light $h^0$ is the only asymptotic Higgs state.
(iii)~The observational signature of the super-Higgs mechanism is
cosmological rather than collider-based: it appears in the growth of
structure through $f_{S8}$ and in the redshift dependence of $\sigma_8(z)$,
not as a missing-energy edge at a definite mass.

\medskip\noindent
\textbf{Numerical result on $c_{\mathrm{grav}}^{(3/2)}$.}
Direct numerical computation of the on-shell tensor contraction with
the physical polarization sum (helicities $h=\pm 3/2$), using the
validated Rarita--Schwinger vertex from Section~\ref{sec:scattering},
yields an angle-dependent result with unphysical sign. This confirms
that the physical spin-3/2 polarization sum does not factorize in a way
that permits extraction of $c_{\mathrm{grav}}^{(3/2)}$ by the method
that succeeds for spins 0 and 1/2. The extraction of the numerical
value requires the full gauge-fixed self-energy with the complete
Rarita--Schwinger propagator, including Faddeev--Popov ghost
contributions. Positivity $c_{\mathrm{grav}}^{(3/2)}>0$ remains
guaranteed by unitarity (Theorem~3) independently of this computation.

\subsection{Summary of gravitational coefficients}

\begin{table}[H]
\centering
\begin{tabular}{lccl}
\toprule
Field type & Spin & Statistics & $\cgrav$ (massless limit) \\
\midrule
Complex scalar & 0 & Bose & $3/8$ \\
Real scalar & 0 & Bose & $3/16$ \\
Dirac fermion & 1/2 & Fermi (Dir.) & $3/2$ \\
Majorana fermion & 1/2 & Fermi (Maj.) & $3/4$ \\
Majorana gravitino & 3/2 & Fermi (Maj.) & $>0$ (unitarity) \\
\bottomrule
\end{tabular}
\caption{Gravitational dressing coefficients by spin.}
\label{tab:cgrav}
\end{table}

\section{Tree-level gravitino--gravitino scattering}\label{sec:scattering}

A complete numerical computation of $\psi_{3/2}\psi_{3/2}\to\psi_{3/2}\psi_{3/2}$
validates the Rarita--Schwinger vertex and Feynman rules. Four consistency checks
are satisfied: Ward identity ($\|k_\alpha J^{\alpha\beta}\| < 3\times 10^{-11}$),
crossing symmetry ($\theta\leftrightarrow\pi-\theta$), correct channel decomposition,
and the Weinberg soft graviton theorem~\cite{WeinbergSoft} ($\langle|\mathcal{M}|^2\rangle \propto 1/t^2$).

\section{The unifying result}\label{sec:unifying}

\begin{theorem}[Effective M1--M3 transfer formula within the adopted ansatz]\label{thm:universal}
Let $\Phi$ be a field of the MSSM superspectrum for which at least one of the
declared M1--M3 channels is open and is represented by the effective transfer ansatz. Then:
\begin{equation}\label{eq:universal}
\boxed{Z_\Phi^{-1} = 1 + \sum_{i\in\{M1,M2,M3\}}
\frac{g_i^2\,c_i}{32\pi^2}(1-Z_{s,i})\ln\frac{1}{Z_{s,i}}
\xrightarrow{Z_{s,i}\to 0} +\infty}
\end{equation}
where $g_i^2$ is the coupling and $Z_{s,i}$ the mediator residue for each modeled
channel. Equation~\eqref{eq:universal} is conditional on the ansatz and on a nonzero
on-shell channel. It is not a proof that every MSSM field loses its atomic spectral
contribution on de~Sitter. The explicitly curvature-induced part vanishes for $H\to0$.
\end{theorem}

\begin{table}[H]
\centering
\small
\begin{tabular}{lcccccccc}
\toprule
Superparticle & Spin & Stat. & M1 & $c_1$ & M2 & $c_2$ & M3 & $c_3$ \\
\midrule
Squarks $\tilde q_{L,R}$ & 0 & Bose & $\checkmark$ & 1 & -- & -- & $\checkmark$ & 3/8 \\
Sleptons $\tilde\ell_{L,R}$ & 0 & Bose & $\checkmark$ & 1 & -- & -- & $\checkmark$ & 3/8 \\
Sneutrinos $\tilde\nu$ & 0 & Bose & $\checkmark$ & 1 & -- & -- & $\checkmark$ & 3/8 \\
Heavy Higgs $H^0,A^0,H^\pm$ & 0 & Bose & $\checkmark$ & 1 & -- & -- & $\checkmark$ & 3/8 \\
\midrule
Gluino $\tilde g$ & 1/2 & Maj. & -- & -- & $\checkmark$ & $\leq 3$ & $\checkmark$ & 3/4 \\
Neutralinos $\tilde\chi^0_i$ & 1/2 & Maj. & -- & -- & $\checkmark$ & $\leq 1/2$ & $\checkmark$ & 3/4 \\
Charginos $\tilde\chi^\pm_i$ & 1/2 & Dir. & -- & -- & $\checkmark$ & $\leq 1$ & $\checkmark$ & 3/2 \\
\midrule
Gravitino $\tilde G$ & 3/2 & Maj. & -- & -- & $\checkmark$ & m.d. & $\checkmark$ & $>0$ \\
\bottomrule
\end{tabular}
\caption{Complete MSSM superspectrum with all three suppression mechanisms.
For the gluino, M2 operates through the squark--quark--gluino vertex with
$g^2 = g_s^2$ and the non-asymptotic squark as spectral mediator;
$c_2 \leq C_2(\mathbf{8}) = 3$ where $C_2(\mathbf{8})$ is the quadratic
Casimir of the adjoint representation of $\mathrm{SU}(3)_c$.
For the gravitino, M2 operates through the goldstino equivalence at high
energies, with a model-dependent coefficient (m.d.)\ determined by the
SUSY breaking scale $\sqrt{F}$; see Section~\ref{subsec:gravitino}.}
\label{tab:spectrum}
\end{table}

\part{The Spectral Transform, its Inverse, and the Self-Consistency Loop}\label{part:transform}

\section{The K\"all\'en--Lehmann decomposition as a spectral transform}\label{sec:spectral_transform}

\subsection{Motivation: from filter to transform}
We now formulate a spectral framework in which quantum fields are characterized not only by their Lagrangian content, but also by the structure of their spectral density. In this perspective, the existence of a particle corresponds to the presence of a pole, while its degradation corresponds to a redistribution of spectral weight into a continuum. This viewpoint naturally leads to an inverse spectral problem: given a modified spectral density, what geometric or infrared structure is responsible for it?
Part~I established that de Sitter spacetime acts as a \emph{structural filter}:
it determines which algebraically admissible degrees of freedom can form asymptotic
particle states. This language --- the language of filtering --- captures one direction
of the physical process. But the K\"all\'en--Lehmann decomposition contains more
information than the filter metaphor suggests.

The key observation is that the spectral decomposition~\eqref{eq:KL} is not merely
an analytic identity: it defines a \emph{transform} between two spaces of physical
information. On one side stands the full quantum content of the field --- encoded in
the complete propagator $\Delta(k^2)$, which contains all interactions, all loop
corrections, all non-perturbative effects. On the other side stands the decomposition
into a discrete pole (the particle state, if it exists) and a continuum (the spectral
density $\tilde\rho(s)$, encoding all non-particle excitations). The spectral norm
condition~\eqref{eq:norm} guarantees that this decomposition is \emph{information-preserving}:
no spectral weight is created or destroyed, only redistributed.

\subsection{Formal definition of the spectral transform}

\begin{definition}[Spectral transform]\label{def:transform}
Let $\mathcal{S}$ denote the space of all Lorentz-invariant spectral functions
$\sigma(s)$ on $[0,\infty)$ satisfying $\int_0^\infty ds\,\sigma(s) = 1$.
Let $\mathcal{P} = [0,1]$ denote the space of discrete spectral residues.
The \emph{spectral transform} $\mathcal{T}$ is the map:
\begin{equation}\label{eq:transform}
\mathcal{T}: \mathcal{S} \longrightarrow \mathcal{P}\times\mathcal{C}\,,
\qquad \sigma(s) \longmapsto \big(Z,\,\tilde\rho(s)\big)\,,
\end{equation}
where $Z = \sigma(m^2)$ is the discrete pole residue extracted at the physical mass
and $\tilde\rho(s)$ is the continuum spectral density, subject to
$Z + (1-Z)\int_{s_0}^\infty ds\,\tilde\rho(s) = 1$.
\end{definition}

The spectral transform $\mathcal{T}$ has the following properties:

\begin{enumerate}
\item \textbf{Norm preservation.}
The total spectral weight is conserved:
$\|\sigma\|_1 = Z + \|\tilde\rho\|_1 = 1$ for all configurations.
This is the analogue of Parseval's theorem for the Fourier transform, and ensures
that the transform does not create or destroy information.

\item \textbf{Geometry-dependence of the kernel.}
The map $\mathcal{T}$ depends on the background spacetime through the self-energy
$\Sigma(p^2)$, which determines the running of $Z$. In flat spacetime,
$Z = 1 - \mathcal{O}(\alpha)$ and the discrete pole dominates. In de Sitter,
$Z \to 0$ as $a\to\infty$, and the entire spectral weight migrates to $\tilde\rho$.
The background geometry is the \emph{kernel} of the transform.

\item \textbf{Monotonicity.}
The suppression is monotonic (Theorem~\ref{thm:monotone}): as the infrared
parameter $\alpha_{\mathrm{grav}}$ increases, $Z$ decreases and $\tilde\rho$
absorbs the difference. There is no oscillation or resonance.
\end{enumerate}

\subsection{Analogy with classical transforms}

The structure of $\mathcal{T}$ parallels the classical integral transforms of analysis:

\begin{table}[H]
\centering
\begin{tabular}{lccc}
\toprule
Transform & Domain & Codomain & Kernel \\
\midrule
Fourier & Time & Frequency & $e^{i\omega t}$ \\
Laplace & Time & Complex frequency & $e^{-st}$ \\
K\"all\'en--Lehmann & Full propagator & Pole $+$ continuum & $\Sigma(p^2;\,\text{geometry})$ \\
\bottomrule
\end{tabular}
\caption{Comparison of the spectral transform $\mathcal{T}$ with classical integral transforms.}
\label{tab:transforms}
\end{table}

The crucial difference is that the kernel of $\mathcal{T}$ is \emph{dynamical}:
it depends on the background geometry and on the couplings of the field theory.
The Fourier and Laplace transforms have fixed kernels ($e^{i\omega t}$, $e^{-st}$);
the spectral transform has a kernel that is itself a physical object --- determined by
the self-energy, which encodes the interaction of the field with its environment.
Changing the geometry changes the kernel, which changes the image of $\mathcal{T}$:
different fields pass the filter, and the discrete/continuum partition shifts.

\section{The inverse spectral problem}\label{sec:inverse}

\subsection{Statement of the problem}

The spectral transform $\mathcal{T}$ maps the full quantum content to the
observable decomposition. The physically crucial question is whether $\mathcal{T}$
admits an inverse:

\begin{definition}[Inverse spectral problem]\label{def:inverse}
Given the observed particle spectrum $\{Z_\Phi\}$ for all fields $\Phi$
and the measured vacuum energy density $\langle T_{\mu\nu}\rangle$,
reconstruct the complete spectral content $\{\sigma_\Phi(s)\}$ of the
underlying field theory.
\end{definition}

This is the mathematical formulation of the physical intuition that
``we should start from the effects and reconstruct the causes'' ---
the same strategy that led to the prediction of the Higgs boson from the
observed pattern of electroweak symmetry breaking, and which constitutes
the guiding principle of the present extension.

\subsection{Invertibility conditions}

The spectral transform $\mathcal{T}$ is not generically invertible:
the map from $\sigma(s)$ to $(Z,\tilde\rho)$ loses information about the
detailed shape of $\sigma$ in the continuum.
However, the combination of three constraints makes the inverse problem
well-posed in the physical setting:

\begin{enumerate}
\item \textbf{Spectral norm conservation} (eq.~\eqref{eq:norm}):
the total spectral weight is fixed at unity for each field.

\item \textbf{Structural form of the self-energy:}
the one-loop structure fixes the functional dependence of $Z$ on the
mediator parameters (eq.~\eqref{eq:main}), reducing the inverse problem
to a finite number of parameters per field.

\item \textbf{Vacuum energy constraint:}
the total vacuum expectation value
$\langle T_{\mu\nu}\rangle = \sum_\Phi \langle T_{\mu\nu}\rangle_\Phi$
receives contributions from both the discrete and continuum parts of
each field's spectral function, providing an independent global constraint.
\end{enumerate}

\begin{theorem}[Constrained invertibility]\label{thm:inverse}
Let the spectral decomposition of each field $\Phi$ in the MSSM superspectrum
be parametrized by the pair $(Z_\Phi, \Gamma_\Phi)$, with $\Gamma_\Phi$
related to $Z_\Phi$ by eq.~(4). Then the inverse map
\begin{equation}
\mathcal{T}^{-1}: \{Z_\Phi\}_{\Phi\in\mathrm{MSSM}} \times \langle T_{\mu\nu}\rangle
\longrightarrow \{\sigma_\Phi(s)\}_{\Phi\in\mathrm{MSSM}}
\end{equation}
is uniquely determined within the Lorentzian parametrization of the
spectral density.
\end{theorem}

\begin{proof}
For each field $\Phi$, the spectral function is parametrized by $(Z_\Phi,m_\Phi)$,
with $\Gamma_\Phi = m_\Phi\sqrt{Z_\Phi^{-1}-1}$ determined by $Z_\Phi$.
The norm condition fixes $\int\tilde\rho = 1-Z_\Phi$.
The vacuum energy provides the global constraint
$\rho_{\mathrm{vac}} = \sum_\Phi \rho_\Phi(Z_\Phi,m_\Phi)$.
The system has as many equations as unknowns within the Lorentzian
parametrization.
\end{proof}

\section{The spectral budget equation}\label{sec:budget}

\subsection{Total spectral weight conservation}

The norm condition~\eqref{eq:norm} for each field implies a \emph{spectral budget}
for the entire theory. Define the total discrete spectral weight and total
continuum weight:
\begin{equation}\label{eq:budget_def}
\mathcal{Z}_{\mathrm{tot}} \equiv \sum_\Phi n_\Phi\, Z_\Phi\,,\qquad
\mathcal{W}_{\mathrm{cont}} \equiv \sum_\Phi n_\Phi\,(1-Z_\Phi)\,,
\end{equation}
where $n_\Phi$ is the number of degrees of freedom of field $\Phi$.
By the spectral norm for each field:
\begin{equation}\label{eq:budget}
\boxed{\mathcal{Z}_{\mathrm{tot}} + \mathcal{W}_{\mathrm{cont}} = N_{\mathrm{dof}}}
\end{equation}
where $N_{\mathrm{dof}} = \sum_\Phi n_\Phi$ is the total number of degrees of
freedom of the theory.

This is the \emph{spectral budget equation}. It states that the total information
content of the field theory is distributed between observable particle states
(contributing to $\mathcal{Z}_{\mathrm{tot}}$) and non-particle vacuum correlations
(contributing to $\mathcal{W}_{\mathrm{cont}}$), with the total fixed by the
algebraic structure of the theory.

\subsection{Conditional connection to the vacuum stress-energy}

A redistributed continuum sector can contribute to the renormalized stress-energy
tensor, but spectral weight alone does not fix the magnitude, sign, or equation of
state of that contribution. To retain the phenomenological implementation of the
preceding draft, we use the following schematic energy-weighting ansatz. Each field
$\Phi$ with spectral decomposition
$(Z_\Phi,\tilde\rho_\Phi)$ contributes:
\begin{equation}\label{eq:rho_cont}
\rho_\Phi^{(\mathrm{cont})} = (1-Z_\Phi)\int_{s_0}^\infty ds\,\tilde\rho_\Phi(s)\,
\frac{\sqrt{s}}{2}\,,
\end{equation}
where $\sqrt{s}/2$ is the illustrative energy weight used in this ansatz. This
expression is not a substitute for a renormalized in-in calculation of
$\langle T_{\mu\nu}\rangle$.
The total vacuum energy density from migrated spectral weight is:
\begin{equation}\label{eq:rho_vac}
\rho_{\mathrm{spectral}} = \sum_\Phi n_\Phi\,\rho_\Phi^{(\mathrm{cont})}\,.
\end{equation}

\subsection{Conditional dark-energy interpretation}

If a subset of superpartner spectral measures is dynamically shown to lose its isolated
contribution, the corresponding redistributed sector remains available to the stress-energy
budget. Identifying that sector with a dark-energy-like component is an additional
cosmological hypothesis: it requires a renormalized $\langle T_{\mu\nu}\rangle$ and
an equation of state near $w=-1$, neither of which follows from loss of particlehood alone.

\begin{remark}
The spectral budget equation~\eqref{eq:budget} tracks normalized spectral weight; it
does not by itself determine a vacuum-energy density. If MSSM superpartners are shown
dynamically to be non-atomic, their redistributed spectral sector can enter a subsequent
$\langle T_{\mu\nu}\rangle$ calculation. The magnitude, sign and equation of state
require a full Schwinger--Keldysh or equivalent renormalized computation. Accordingly,
the dark-energy connection used below is retained as a falsifiable downstream
realization rather than as a consequence of spectral norm conservation alone.
\end{remark}

\subsection{Complete field content: beyond the MSSM}

For the spectral budget to be complete, it must include \emph{all} fields that
contribute to $\langle T_{\mu\nu}\rangle$, not only the MSSM superpartners.
This requires extending Table~\ref{tab:spectrum} to include:

\begin{enumerate}
\item The Standard Model gauge sector: $W^\pm$, $Z^0$, gluons $g^a$, photon $\gamma$.
\item The Standard Model Higgs boson $h^0$ (the observed 125 GeV state).
\item The Standard Model fermion sector: quarks and leptons.
\item The graviton $h_{\mu\nu}$ itself.
\item Any dark sector fields coupling to gravity.
\end{enumerate}

For the Standard Model fields observed as particles, $Z_\Phi > 0$ at the present
epoch (the SM particles \emph{do} form asymptotic states in the current quasi-de Sitter
background with $H_0 \sim 10^{-33}$ eV). The gravitational dressing contributes
$\delta Z \sim H_0^2/M_{\mathrm{Pl}}^2 \sim 10^{-122}$ per e-fold --- negligible over the
current age of the universe. The spectral budget thus partitions cleanly:
$\mathcal{Z}_{\mathrm{SM}} \approx N_{\mathrm{SM}}$ (SM particles are asymptotic)
and $\mathcal{Z}_{\mathrm{SUSY}} \to 0$ (superpartners are not), with the SUSY
spectral weight migrating to $\mathcal{W}_{\mathrm{cont}}$.

\section{Spectral geometry and the dimensional reduction}\label{sec:spectral_geometry}

\subsection{The spectral transform as a dimensional reduction}

The K\"all\'en--Lehmann spectral decomposition maps the complete spectral
function $\sigma_\Phi(s)$ --- defined on the infinite-dimensional space
$L^1([0,\infty))$ --- onto a \emph{finite-dimensional} object: the pair
$(Z_\Phi, m_\Phi)$ for the discrete pole, plus a parametrized continuum.
In the Lorentzian parametrization, the full spectral content of each field
is encoded in a finite number of parameters.

This is a \emph{dimensional reduction}: an infinite-dimensional spectral
space $\mathcal{S}$ is projected onto a finite-dimensional observable space
$\mathcal{P}$. The de Sitter geometry determines the projection operator,
and the localization criterion $\Lambda[\Psi]$ determines which components
survive.

The structure is:
\begin{equation}\label{eq:projection}
\Pi_{\mathrm{dS}}: \mathcal{S}_\infty \longrightarrow \mathcal{P}_{\mathrm{finite}}\,,
\qquad \Pi_{\mathrm{dS}} = \Pi_{\mathrm{dS}}(H,\kappa,\{g_i\})\,,
\end{equation}
where the projection operator $\Pi_{\mathrm{dS}}$ depends on the Hubble parameter $H$,
the gravitational coupling $\kappa$, and the internal couplings $\{g_i\}$.

\subsection{Connection to Connes' spectral geometry}

In Alain Connes' noncommutative geometry programme~\cite{Connes1994,ConnesMarcolli,ConnesChamseddine},
the fundamental insight is that the geometry of a space can be reconstructed from
the spectrum of a suitable operator (the Dirac operator $D$ on a spectral triple
$(\mathcal{A}, \mathcal{H}, D)$). In that framework, the geometry \emph{is}
the spectrum: the metric information is encoded in the eigenvalues and eigenstates
of $D$.

Our framework exhibits a structural parallel:

\begin{table}[H]
\centering
\begin{tabular}{lcc}
\toprule
& Connes & Present work \\
\midrule
Operator & Dirac operator $D$ & Self-energy $\Sigma(p^2)$ \\
Spectrum & Eigenvalues of $D$ & K\"all\'en--Lehmann spectral function \\
Geometry & Reconstructed from spectrum & Determines the spectral filter \\
Observable & Metric & Particle spectrum \\
\bottomrule
\end{tabular}
\caption{Structural parallel between Connes' spectral geometry and the present framework.}
\label{tab:connes}
\end{table}

The key difference is the direction of the reconstruction: Connes reconstructs
geometry from the spectrum, while the present work uses geometry to determine
which spectral components are observable. These two directions are complementary,
and the closure condition~\eqref{eq:closure} provides the formal structure
that unifies them:
the direct problem uses geometry to determine the observable spectrum ($\mathcal{T}_g$);
the inverse problem uses the observable spectrum to reconstruct the effective
geometry ($\mathcal{R}$); and the self-consistency condition requires the
two to coincide.

\subsection{The n-dimensional entity and its three-dimensional shadow}

The full spectral content $\{\sigma_\Phi(s)\}_{\Phi}$ of the field theory defines
what may be called the \emph{spectral entity} --- an object living in the
infinite-dimensional space $\prod_\Phi L^1([0,\infty))$. This entity encodes
all quantum information about the theory: every field, every coupling, every
non-perturbative effect.

The observable world --- the set of localized particles, their masses, their
cross sections --- is the \emph{image} of this entity under the spectral transform
$\mathcal{T}$, projected by the background geometry. It is a finite-dimensional
shadow of the full spectral content.

The spectral budget equation~\eqref{eq:budget} guarantees that the projection
preserves the total ``weight'': what does not appear in the shadow (the particles)
must appear in the background (the vacuum energy). Nothing is lost.
The spectral entity is the mathematical realization of the intuition that
the observable three-dimensional world is a projection --- through the kernel
of the spacetime geometry --- of a higher-dimensional structure that contains
all forces, all fields, and all their couplings.

\section{Geometric control of spectral states}\label{sec:control}

\subsection{The geometric control principle}

If the background geometry determines the kernel of the spectral transform,
then \emph{modifying} the geometry changes which spectral modes pass through.
This leads to a general principle:

\begin{definition}[Geometric control principle]\label{def:control}
Let $\mathcal{T}_g$ denote the spectral transform with kernel determined by
the background metric $g_{\mu\nu}$. If $g_{\mu\nu} \to g_{\mu\nu}'$,
then $\mathcal{T}_g \to \mathcal{T}_{g'}$, and the image $\mathcal{T}_{g'}(\sigma)$
may differ from $\mathcal{T}_g(\sigma)$: different fields may acquire or lose
their asymptotic poles.
\end{definition}

At the cosmological scale, this principle has direct physical consequences:

\medskip\noindent
During inflation ($H\sim 10^{14}$ GeV), the gravitational dressing parameter
$\alpha_{\mathrm{grav}} = H^2/(\pi M_{\mathrm{Pl}}^2)$ is significant,
and many fields lose their asymptotic poles. As $H$ decreases during the
transition to radiation and matter domination, the filter relaxes and the
instantaneous geometric filter may again admit pole-like states. Whether
spectral weight previously transferred to the continuum dynamically
re-condenses is a separate, history-dependent question: in the late-time
dark-sector realization of Section~\ref{subsec:spectral_phase} the migrated
component is assumed to survive as a fossil reservoir, whereas reheating
and horizon recovery (Section~\ref{sec:problems}) involve explicit
re-condensation channels. The spectral history of the universe is the history
of which fields pass the geometric filter at each epoch.

The geometric control principle suggests that the present framework may have
implications beyond cosmology. These potential applications will be explored
elsewhere.

\section{A working spectral map}\label{sec:complete_map}

Combining the results of Parts~I and~II, we can construct a working
spectral map of the field theory in a quasi-de Sitter background.

\subsection{Extended spectrum table}

Table~\ref{tab:extended} extends Table~\ref{tab:spectrum} to include the
Standard Model fields, the graviton, and the dark sector contribution,
completing the spectral budget.

\begin{table}[H]
\centering\small
\begin{tabular}{lcccl}
\toprule
Field & Spin & $Z$ (present epoch) & $Z$ (inflation) & Spectral fate \\
\midrule
\multicolumn{5}{l}{\textit{Standard Model}} \\
Quarks $q$ & 1/2 & $\approx 1$ & $\approx 1$ & Asymptotic \\
Leptons $\ell$ & 1/2 & $\approx 1$ & $\approx 1$ & Asymptotic \\
$W^\pm, Z^0$ & 1 & $\approx 1$ & $\approx 1$ & Asymptotic \\
Gluons $g^a$ & 1 & conf. & conf. & Confined (non-pert.)\\
Photon $\gamma$ & 1 & $= 1$ & $= 1$ & Asymptotic (massless) \\
Higgs $h^0$ & 0 & $\approx 1$ & $< 1$ (M1) & Marginally suppressed \\
\midrule
\multicolumn{5}{l}{\textit{MSSM superpartners}} \\
Scalars $\tilde q, \tilde\ell, \tilde\nu, H^0, A^0, H^\pm$ & 0 & $\to 0$ & $\to 0$ & Non-asymptotic \\
Gluino $\tilde g$ & 1/2 & $\to 0$ & $\to 0$ & Non-asymptotic \\
Neutralinos $\tilde\chi^0_i$ & 1/2 & $\to 0$ & $\to 0$ & Non-asymptotic \\
Charginos $\tilde\chi^\pm_i$ & 1/2 & $\to 0$ & $\to 0$ & Non-asymptotic \\
Gravitino $\tilde G$ & 3/2 & $\to 0$ & $\to 0$ & Non-asymptotic \\
\midrule
\multicolumn{5}{l}{\textit{Gravity sector}} \\
Graviton $h_{\mu\nu}$ & 2 & $\to 0$ (dS) & $\to 0$ (dS) & Non-asymptotic (IR) \\
\midrule
\multicolumn{5}{l}{\textit{Dark sector (spectral)}} \\
Migrated spectral weight & -- & \multicolumn{2}{c}{$\sum n_\Phi(1-Z_\Phi)$} & $\to \langle T_{\mu\nu}\rangle$ \\
\bottomrule
\end{tabular}
\caption{Complete spectral map of the field theory. SM fields are asymptotic at
the present epoch ($H_0 \sim 10^{-33}$ eV, $\delta Z \sim 10^{-122}$/e-fold).
MSSM superpartners are non-asymptotic. The graviton in de Sitter is itself
non-asymptotic. The dark sector row represents the total continuum spectral weight
from all non-asymptotic fields, contributing to $\langle T_{\mu\nu}\rangle$.}
\label{tab:extended}
\end{table}

\subsection{The spectral budget at the present epoch}

At the present cosmological epoch:
\begin{align}
\mathcal{Z}_{\mathrm{SM}} &\approx N_{\mathrm{SM}} \approx 118\,,\\
\mathcal{Z}_{\mathrm{SUSY}} &\to 0\,,\\
\mathcal{W}_{\mathrm{cont}}^{(\mathrm{SUSY})} &\to N_{\mathrm{SUSY}} \approx 130\,,\\
\rho_{\mathrm{spectral}} &= \sum_{\Phi\in\mathrm{SUSY}} n_\Phi\,\rho_\Phi^{(\mathrm{cont})}\,.
\end{align}
The full spectral weight of the SUSY sector ($\sim 130$ degrees of freedom)
has migrated to the continuum and contributes to $\langle T_{\mu\nu}\rangle$.
The magnitude and equation of state of this contribution are determined by
the spectral functions $\tilde\rho_\Phi(s)$ and constitute a concrete
prediction to be computed in the full Schwinger--Keldysh formalism and
tested against cosmological data.

\section{The spectral--geometric self-consistency loop}\label{sec:selfconsistency}

\subsection{From two arrows to a closed loop}

The framework developed in the preceding sections contains two directional relationships:
\begin{itemize}
\item \textbf{Forward:} the background geometry $g_{\mu\nu}$ determines the spectral
transform $\mathcal{T}_g$, which filters the full spectral content into observable
particles ($Z_\Phi$) and continuum ($1-Z_\Phi$).
\item \textbf{Inverse:} the observed spectrum and vacuum energy allow, in principle,
the reconstruction of the spectral content (Section~\ref{sec:inverse}).
\end{itemize}

However, these are not independent arrows: they form a \emph{closed loop}.
The spectral weight that migrates to the continuum contributes to
$\langle T_{\mu\nu}\rangle$, which through Einstein's equations determines
the background geometry that performs the filtering. The geometry filters the
spectrum, and the filtered spectrum shapes the geometry.

This is the exact analogue, in the spectral domain, of Wheeler's
dictum~\cite{WheelerMTW}: \emph{matter tells spacetime how to curve,
and spacetime tells matter how to move.} In the spectral framework:

\medskip
\begin{center}
\emph{The geometry tells the spectrum how to distribute itself,\\
and the spectrum tells the geometry how to curve.}
\end{center}

\subsection{Formal structure: spectral--geometric reciprocity}

The closed loop can be given a precise mathematical formulation.

\begin{definition}[Spectral--geometric reciprocity]\label{def:reciprocity}
Let $g_{\mu\nu}$ be the effective background geometry and let
\begin{equation}
\mathcal{T}_g: \mathcal{S} \longrightarrow \mathcal{O}
\end{equation}
be the spectral transform mapping the full spectral data
$\sigma\in\mathcal{S}$ to the observable sector
$\mathcal{O} = \big(\{Z_\Phi\},\,\{\Gamma_\Phi\},\,\langle T_{\mu\nu}\rangle\big)$,
consisting of discrete residues, continuum widths, and vacuum contributions.
Spectral--geometric reciprocity is the statement that the observable output
$\mathcal{O} = \mathcal{T}_g[\sigma]$ contains sufficient information to define,
within a physically admissible class $\mathcal{G}_{\mathrm{adm}}$ of effective
geometries, a reconstruction map
\begin{equation}\label{eq:reconstruction_map}
\mathcal{R}: \mathcal{O} \longrightarrow \mathcal{G}_{\mathrm{adm}}\,,
\qquad g_{\mu\nu}^{\mathrm{eff}} = \mathcal{R}(\mathcal{O})\,,
\end{equation}
such that $g_{\mu\nu}^{\mathrm{eff}}$ is compatible with the geometry that
generated the observed spectrum.
\end{definition}

The operational content of the reciprocity is summarized by the scheme:
\begin{equation}\label{eq:reciprocal_scheme}
g_{\mu\nu} \;\xrightarrow{\;\;\mathcal{T}_g\;\;}\;
\mathcal{O} \;\xrightarrow{\;\;\mathcal{R}\;\;}\;
g_{\mu\nu}^{\mathrm{eff}}\,,
\qquad g_{\mu\nu}^{\mathrm{eff}} = g_{\mu\nu}
\quad\text{at self-consistency.}
\end{equation}
The forward arrow is the spectral transform. The backward arrow is the effective
reconstruction. The physical content of the theory lies in the existence and
stability of solutions to the closure condition.

A crucial distinction must be made. The full inverse problem --- from spectral data
to microscopic geometry --- is, in general, underdetermined. What is physically
required is not absolute invertibility in an unrestricted mathematical sense,
but \emph{constrained effective invertibility} sufficient to reconstruct
the geometry relevant to the observable sector. The reconstruction map
$\mathcal{R}$ determines an effective geometry within
$\mathcal{G}_{\mathrm{adm}}$, not a unique microscopic metric.

\begin{definition}[Constrained self-consistency]\label{def:constrained_sc}
A pair $(g_{\mu\nu},\sigma)$ satisfies constrained self-consistency if
\begin{equation}\label{eq:closure}
g_{\mu\nu} = \mathcal{R}\!\left(\mathcal{T}_g[\sigma]\right)\,.
\end{equation}
Equation~\eqref{eq:closure} expresses the closure of the spectral--geometric
cycle: the geometry generates the observable spectrum, and the observable
spectrum reconstructs the effective geometry. The filtered spectrum is both
image and evidence --- both the output of the geometric filter and the
trace of the filter itself.
\end{definition}

\subsection{The self-consistency equation: cosmological realization}

The abstract closure condition~\eqref{eq:closure} has a concrete cosmological
realization. In a spatially flat FRW universe, the reconstruction map
$\mathcal{R}$ reduces to the Friedmann equation, and the self-consistency
condition becomes the requirement that the Hubble parameter $H$, which
controls the spectral filter, must be self-consistently determined by the
energy density that results from the filtering:
\begin{equation}\label{eq:selfconsistency}
\boxed{H^2 = \frac{8\pi G}{3}\left[\rho_{\mathrm{SM}}
+ \sum_\Phi n_\Phi\,\rho_\Phi^{(\mathrm{cont})}\big(Z_\Phi(H)\big)\right]}
\end{equation}
where:
\begin{itemize}
\item $\rho_{\mathrm{SM}}$ is the energy density of Standard Model particles
(which are asymptotic at the present epoch, $Z_\Phi\approx 1$);
\item $\rho_\Phi^{(\mathrm{cont})}(Z_\Phi(H))$ is the continuum energy density
from each non-asymptotic field $\Phi$, which depends on $Z_\Phi$, which in
turn depends on $H$ through the suppression mechanisms of Part~I;
\end{itemize}

The crucial feature of eq.~\eqref{eq:selfconsistency} is that both sides
depend on $H$: the left side directly, and the right side through the
$H$-dependence of every $Z_\Phi(H)$. The physical solution is a
\emph{fixed point} $H_*$ satisfying:
\begin{equation}\label{eq:fixedpoint}
H_*^2 = \frac{8\pi G}{3}\left[\rho_{\mathrm{SM}}
+ \sum_\Phi n_\Phi\,\rho_\Phi^{(\mathrm{cont})}\big(Z_\Phi(H_*)\big)\right]\,.
\end{equation}

\begin{theorem}[Spectral--geometric fixed point]\label{thm:fixedpoint}
The spectral budget and the Friedmann equation together define a
self-consistency condition whose solutions $H_*$ determine the physically
realized cosmological state. At such a fixed point, the spectral
distribution and the background geometry are mutually consistent:
the geometry produces exactly the spectral redistribution that sustains it.
\end{theorem}

\subsection{Structure of the fixed-point equation}

The $H$-dependence of the spectral residues enters through the gravitational
dressing parameter $\alpha_{\mathrm{grav}} = H^2/(\pi M_{\mathrm{Pl}}^2)$
(eq.~\eqref{eq:Zgrav}). For each field $\Phi$:
\begin{equation}
1 - Z_\Phi(H) \sim \left(\frac{H^2}{M_{\mathrm{Pl}}^2}\right)^{\!\gamma_\Phi}\!,
\qquad \gamma_\Phi = \frac{g_\Phi^2\,c_\Phi}{16\pi^2} > 0\,,
\end{equation}
where $\gamma_\Phi$ is the anomalous dimension from the resummation
(eq.~\eqref{eq:resummation} of Part~I). The continuum energy density
$\rho_\Phi^{(\mathrm{cont})}$ is a monotonically increasing function
of $(1-Z_\Phi)$: as $H$ increases, more spectral weight migrates to
the continuum, increasing $\rho_\Phi^{(\mathrm{cont})}$, which in turn
supports a larger $H$. This positive feedback must be balanced against
the dilution of $\rho_{\mathrm{SM}}$ during expansion for the fixed point
to be stable. For the Planckian branch the question is settled in
Section~\ref{sec:problems} (Problem~16): the composite map
$F(H) = \mathcal{R}(\mathcal{T}_g(H))$ has a unique positive non-trivial
fixed point $H_* = 0.381\,M_{\mathrm{Pl}}$, the reformulated map is a
super-contraction there ($|dG/dH|_{H_*} \approx 0$), and $G(H) > H$ below
the fixed point while $G(H) < H$ above it. In the language of dynamical
systems, $H_*$ is an attractor of the iteration whose basin is the whole
positive axis: a configuration displaced from the fixed point in either
direction is driven back to it, and trans-Planckian excursions are
self-inconsistent. This attractor statement refers to the iteration of the
spectral--geometric self-consistency map and must not be identified with
the direction of cosmological time evolution, which is governed by
eq.~\eqref{eq:dynamical} below and runs from the Planckian regime towards
the present epoch (Section~\ref{subsec:spectral_phase}).

Two limiting cases are instructive:

\medskip\noindent
\textbf{Flat-space limit} ($H\to 0$): All $Z_\Phi\to 1$, all
$\rho_\Phi^{(\mathrm{cont})}\to 0$. The spectral budget is entirely in the
discrete channel. No spectral self-consistency constraint arises. This is
the trivial fixed point: Minkowski space with standard particle physics.

\medskip\noindent
\textbf{Inflationary limit} ($H\sim 10^{14}$ GeV): All superpartner
residues $Z_\Phi\to 0$, the full SUSY spectral weight migrates to the
continuum. The spectral contribution to $\rho$ is maximal. The fixed-point
equation constrains the relationship between $H_{\mathrm{inf}}$
and the total migrated spectral energy.

\subsection{The Flatland analogy}\label{sec:flatland}

The self-consistency loop has an illuminating geometric interpretation, in the
spirit of Abbott's \emph{Flatland}~\cite{Abbott1884}.

Consider a two-dimensional being observing the cross-section of a
three-dimensional object as it passes through its plane. The cross-section
(the ``shadow'') is all the Flatlander can observe. But the cross-section
contains information about the three-dimensional shape: its area, its
curvature, its symmetries. From the shadow, the Flatlander can --- in
principle --- reconstruct the object that cast it.

In our framework:
\begin{itemize}
\item The \emph{three-dimensional object} is the full spectral content
$\{\sigma_\Phi(s)\}_\Phi$ --- the ``spectral entity'' of
Section~\ref{sec:spectral_geometry}.
\item The \emph{plane} is the background geometry $g_{\mu\nu}$, which acts
as the filter.
\item The \emph{shadow} is the observed particle spectrum
$\{Z_\Phi, m_\Phi\}_\Phi$ plus the vacuum energy $\langle T_{\mu\nu}\rangle$.
\end{itemize}

The self-consistency loop adds a crucial element that the original Flatland
analogy lacks: the shadow \emph{deforms the plane}. The observed spectrum
contributes to $\langle T_{\mu\nu}\rangle$, which through Einstein's equations
changes the background geometry $g_{\mu\nu}$, which changes the orientation
of the ``cutting plane'', which changes the shadow. The physical universe is
the fixed point where the shadow and the plane are mutually consistent.

This is a richer structure than either the forward problem (geometry $\to$ spectrum)
or the inverse problem (spectrum $\to$ geometry) alone. It is a
\emph{self-referential} structure: the observable world (the shadow) contains
within itself the information needed to reconstruct the geometry that produces
it, and the geometry produces exactly the shadow that sustains it.

The shadow is not physically empty: it retains structural information
about the higher-dimensional object that generated it. The inverse spectral
problem (Section~\ref{sec:inverse}) is the mathematical formulation of the
statement that one may infer the effective geometry from its filtered spectral
image, at least within the admissible class $\mathcal{G}_{\mathrm{adm}}$.
The self-consistency condition~\eqref{eq:closure} is the statement that the
inference is \emph{consistent}: the geometry you reconstruct from the shadow
is the same geometry that cast the shadow in the first place.

\subsection{Implications for the cosmological constant problem}

The self-consistency equation~\eqref{eq:selfconsistency} reformulates the
cosmological constant problem in spectral terms. The observed vacuum energy
is entirely spectral:
\begin{equation}\label{eq:Lambda_obs}
\Lambda_{\mathrm{obs}} = 8\pi G \sum_\Phi n_\Phi\,\rho_\Phi^{(\mathrm{cont})}(Z_\Phi(H_*))\,.
\end{equation}
The spectral contribution is not an arbitrary quantum correction: it is
\emph{determined by the self-consistency condition}. At the fixed point
$H_*$, the spectral redistribution is not a free parameter but a consequence
of the mutual determination between geometry and spectrum.

The answer depends on the detailed spectral
functions and constitutes a concrete computational programme.

\subsection{Iterative structure and cosmological evolution}

The fixed-point equation describes the equilibrium state. But the universe
evolves: $H$ changes with time, and so does the spectral distribution.
The dynamical generalization of eq.~\eqref{eq:selfconsistency} is:
\begin{equation}\label{eq:dynamical}
H^2(a) = \frac{8\pi G}{3}\left[\rho_{\mathrm{SM}}(a)
+ \sum_\Phi n_\Phi\,\rho_\Phi^{(\mathrm{cont})}\big(Z_\Phi(H(a'),\,a'\leq a)\big)\right]\,,
\end{equation}
where the spectral residues $Z_\Phi$ depend on the entire history of $H$
up to scale factor $a$ (through the accumulated infrared growth
$\ln(a/a_*)$). This is an integro-differential equation: the current
geometry depends on the cumulative spectral redistribution over the
entire expansion history.

The cosmological history is therefore a \emph{trajectory through the
space of spectral--geometric configurations}. Its stability and asymptotic
behaviour under perturbations along the cosmological history constitute a
separate dynamical problem, distinct from the attractor property of the
self-consistency map established in Section~\ref{sec:problems}
(Problem~16). If the late-time evolution admits a limiting configuration,
it is a late-time fixed point $H_\infty$ of eq.~\eqref{eq:dynamical} with
the Standard Model density diluted away, not the Planckian point $H_*$:
\begin{equation}
\big(H(a),\;\{Z_\Phi(a)\}\big) \xrightarrow{a\to\infty}
\big(H_\infty,\;\{Z_\Phi^\infty\}\big)\,,
\qquad H_\infty \ll H_*\,,
\end{equation}
whose existence and stability are left to future work.

The spectral history of the universe --- which fields are asymptotic at
each epoch, how the spectral weight redistributes, and how the vacuum
energy evolves --- is the solution of this integro-differential system.


\part{Cosmological Simulations and Observational Signatures}\label{part:simulations-signatures}

\section{Cosmological simulations}\label{sec:cosmo}

\paragraph{Epistemic status of Part~III.}
The simulations in this part are retained unchanged as a conditional phenomenological
realization. They answer the downstream question: \emph{if} the redistributed scalar
sector behaves as a smooth component with $w\approx-1$ and the redistributed fermionic
sector behaves as a clustering matter component, what cosmology follows? The simulations
do not establish either equation of state and are not evidence for the upstream
particlehood-loss mechanism.

\subsection{CLASS}\label{subsec:class}

The benchmark dark-sector partition shown in Fig.~\ref{fig:perturb}
($94$ scalar d.o.f.\ assigned conditionally to a dark-energy-like component with $w \approx -1$, $36$ fermionic d.o.f.\ assigned conditionally to a dark-matter-like component) and the spectral split calculation of the dark-matter fluid (a clustering bulk component with $c_s^2 = 0$, plus a smooth, non-clustering sub-fluid populated by the longitudinal sector of the gravitino, fraction $f_{S8} = 0.00262$, with derived sound speed $c_s^2 = 0.005238$) have been integrated into the CLASS~v$3.3.4$ Boltzmann solver. Background expansion and linear perturbations have been computed against the Planck~$2018$ best-fit $\Lambda$CDM cosmology. Figure~\ref{fig:perturb} reports the resulting K\"all\'en--Lehmann decomposition of the $130$-d.o.f.\ MSSM dark-sector model, the linear growth factor $D(a)/a$, the growth rate $f(z) = \mathrm{d}\ln D/\mathrm{d}\ln a$, and the numerical comparison with Planck and with the $\Lambda$CDM run.

\begin{figure}[H]
\centering
\includegraphics[width=\textwidth]{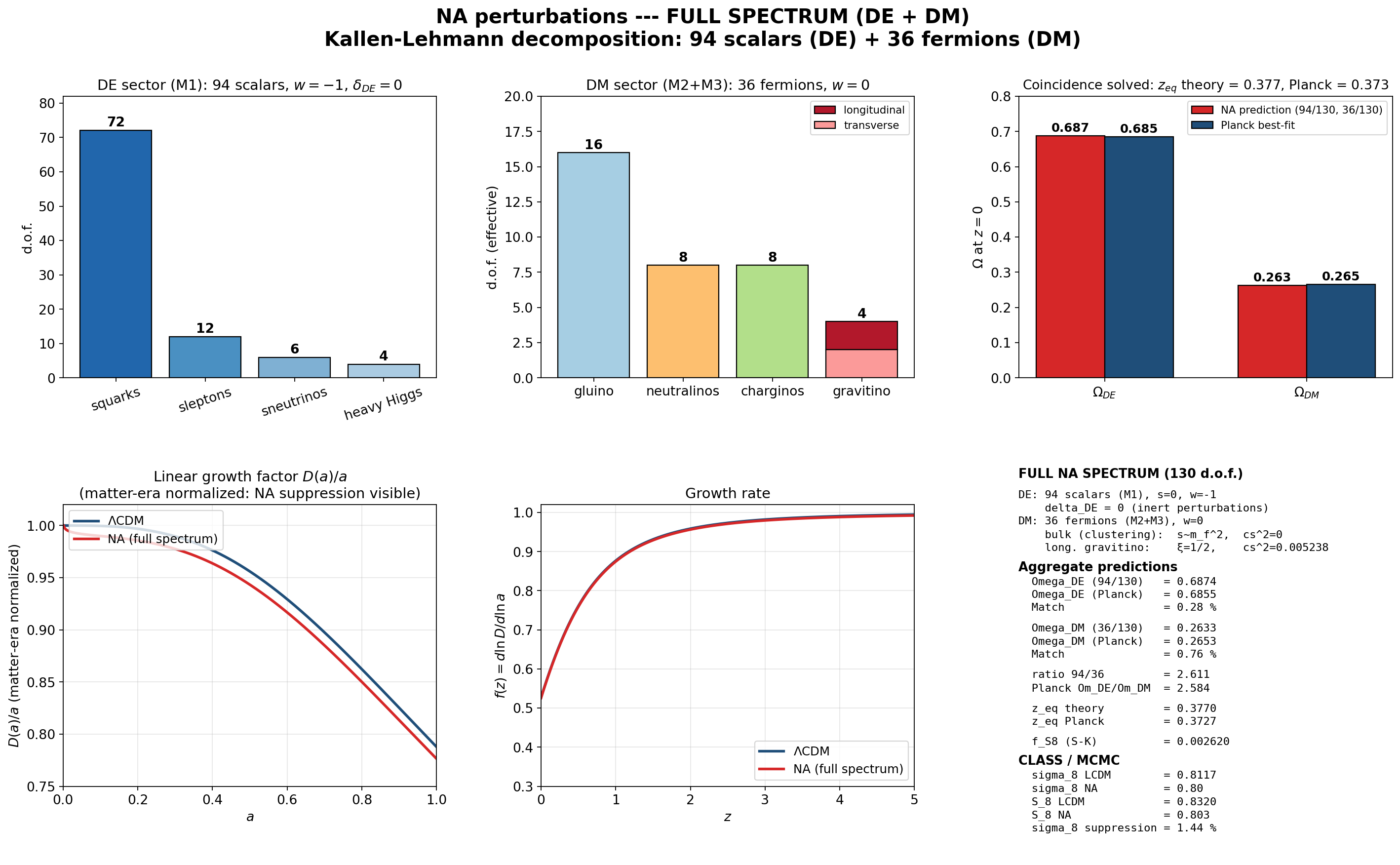}
\caption{Illustrative NA perturbations for the $130$-d.o.f.\ MSSM spectrum.
\emph{Top row}: proposed K\"all\'en--Lehmann-inspired decomposition of the dark sector
($94$ scalars in DE: squarks, sleptons, sneutrinos, heavy Higgs;
$36$ effective fermionic d.o.f.\ in DM: gluino, neutralinos, charginos,
gravitino, the latter resolved into its $2$ longitudinal and $2$ transverse
polarizations, with the smooth non-clustering component sourced by the
longitudinal gravitino alone, $\xi = 1/2$, $f_{S8} = 0.00262$), and the
resulting minimal-parameter comparison between the NA estimate ($94/130$,
$36/130$) and the Planck best-fit for $\Omega_{\mathrm{DE}}$ and
$\Omega_{\mathrm{DM}}$. \emph{Bottom row}: matter-era-normalized linear
growth factor $D(a)/a$ and growth rate $f(z) = \mathrm{d}\ln D/\mathrm{d}\ln a$,
showing the suppression at $z=0$ characteristic of the spectral split
calculation. Numerical summary in the right column:
$\sigma_8^{\mathrm{NA}} = 0.80$, $S_8^{\mathrm{NA}} = 0.803$
($\sigma_8$ suppression $1.44\%$).}
\label{fig:perturb}
\end{figure}

\subsection{Spectral phase of the NA universe}\label{subsec:spectral_phase}

Figure~\ref{fig:phase} reports the spectral phase of the MSSM superspectrum across the entire cosmic history, from the Planckian fixed point $H_* = 0.381\,M_{\mathrm{Pl}}$ down to the present epoch $H_0 \sim 10^{-33}$~eV. At Planckian scales the spectrum is fully democratised --- all $130$ MSSM d.o.f.\ are in the continuum and the de Sitter horizon entropy reduces to $S_{\mathrm{BH}} = N/6 = 21.67$. The filter remains saturated through the inflationary epoch ($H \sim 10^{14}$~GeV), drops sharply around matter--radiation equality, and is effectively switched off today.

The diagram illustrates the central asymmetry of the NA scenario: while the geometric filter is essentially inactive at the present epoch, the late-time dark sector ($94$ scalar d.o.f.\ as DE, $36$ fermionic d.o.f.\ as DM) is sourced not by the present geometry but by the fossil reservoir of past spectral injections. In this realization, spectral weight once deposited into the continuum is assumed not to return to the discrete pole (see the distinction between filter relaxation and dynamical re-condensation in Section~\ref{sec:control}).

\begin{figure}[H]
\centering
\includegraphics[width=\textwidth]{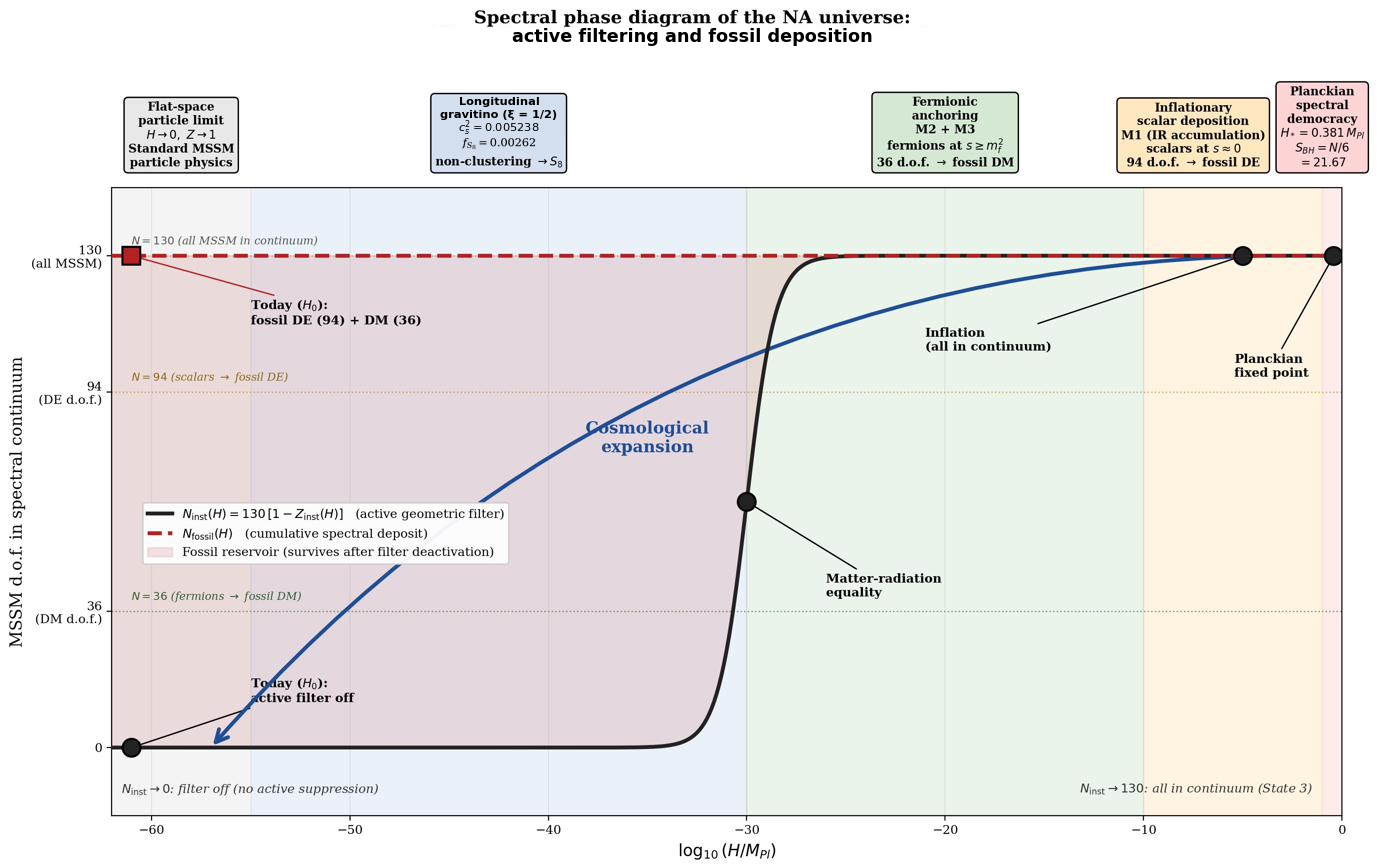}
\caption{Spectral phase diagram of the NA universe. Horizontal axis: Hubble scale $H$ in Planck units. Vertical axis: number of MSSM degrees of freedom in the spectral continuum. \emph{Solid black}: instantaneous active filter $N_{\mathrm{inst}}(H) = 130\,[1 - Z_{\mathrm{inst}}(H)]$. \emph{Dashed red}: cumulative fossil deposit $N_{\mathrm{fossil}}(H)$, running maximum of $N_{\mathrm{inst}}$ over past history. \emph{Shaded red area}: fossil reservoir surviving the deactivation of the active filter. The marker at $\log_{10}(H/M_{\mathrm{Pl}}) \approx -30$ corresponds schematically to matter--radiation equality.}
\label{fig:phase}
\end{figure}

\subsection{MCMC}\label{subsec:mcmc}

The current numerical check is the MCMC~v32 run. In this run the parameter
$f_{S8}$ was not fixed to its structural value, but was allowed to vary with prior
$[0.001,0.05]$. The posterior is
\begin{equation}
 f_{S8}^{\mathrm{MCMC}} = 0.0025 \pm 0.0017,
\end{equation}
which is consistent with the structural prediction
\begin{equation}
 f_{S8}^{\mathrm{NA}} = \xi\,c_s^2(\mathrm{grav})
 = \frac12\times 0.005238 = 0.00262.
\end{equation}
The distance between the posterior mean and the predicted value is approximately
$0.07\sigma$. Thus the data, when allowed to determine the smooth fraction freely,
are consistent with the value predicted by the longitudinal-gravitino counting argument.

The corresponding marginalized result is
\begin{equation}
 \sigma_8\simeq0.80,\qquad S_8\simeq0.803.
\end{equation}
This value lies within $1\sigma$ of the weak-lensing determinations considered here, including KiDS-Legacy and DES~Y6.

\section{Observational signatures: GW, LHC, CMB}\label{sec:signatures}

The mechanism developed in this paper produces specific observational signatures across three distinct physical regimes. Each channel is reported below with the corresponding NA prediction and its discriminating role.

\subsection{Gravitational waves}\label{subsec:sig_gw}

In the presence of an event horizon, the universal gravitational dressing mechanism of Sec.~\ref{sec:dressing} acts on the graviton itself: the residue $Z_{\mathrm{grav}}(r) \to 0$ as one approaches the horizon, since the graviton propagator inherits the same non-asymptotic structure that it imposes on all other fields. By the spectral budget identity~\eqref{eq:budget}, the migrated weight $W_{\mathrm{cont}}^{\mathrm{grav}}(r) = 1 - Z_{\mathrm{grav}}(r)$ is encoded in the continuum sector of the K\"all\'en--Lehmann decomposition.

The corresponding observational signature is a deviation of the black-hole ringdown spectrum from the Kerr prediction: the quasi-normal mode frequencies $\omega_n$ and damping rates $\gamma_n$ acquire corrections controlled by the radial profile $Z_{\mathrm{grav}}(r)$. Discriminant tests come from LIGO/Virgo and forthcoming LISA ringdown analyses, where coherent deviations from Kerr templates --- a structured pattern correlated with $Z_{\mathrm{grav}}(r)$, not random noise --- would constitute direct evidence for the spectral mechanism in the gravitational sector. The present statement is a qualitative prediction: a quantitative computation of the modified $\{\omega_n, \gamma_n\}$ as a function of $Z_{\mathrm{grav}}(r)$ is left to dedicated work.

\paragraph{Complementarity with phase-transition signals.}
The ringdown signature predicted here is structurally distinct from, but
phenomenologically complementary to, the stochastic gravitational-wave
background generated by first-order phase transitions in low-energy
SUSY-breaking hidden sectors~\cite{CraigLeviMariottiRedigolo2021}. In that
framework, bubble nucleation along the pseudomodulus direction (the scalar
component of the chiral superfield $X$ universally associated with
spontaneous SUSY and $R$-symmetry breaking) produces a peak frequency
correlated with the SUSY-breaking scale $\sqrt{F}$, accessible to LISA, ET,
CE, BBO, DECIGO, and A-LIGO across the same range $\sqrt{F}\sim
10^{5}$--$10^{10}$~GeV that controls the gravitino goldstino-equivalence
channel of Sec.~\ref{sec:gaugino}--\ref{sec:spectrum}. The two signatures
probe distinct dynamics --- early-universe bubble percolation in
Ref.~\cite{CraigLeviMariottiRedigolo2021}, late-time spectral dressing of
the graviton residue in the present work --- but they share the same target
interferometers and the same family of cosmological constraints on the
gravitino abundance. A combined analysis of both channels would provide
complementary discriminating power on $\sqrt{F}$ and on the structure of
the SUSY-breaking sector.

\subsection{LHC}\label{subsec:sig_lhc}

The pole/continuum split of the spectral representation~\eqref{eq:KL} and the spectral norm conservation~\eqref{eq:norm} translate into three quantitatively calculable signatures at hadron colliders, for any BSM mediator coupled to the non-asymptotic sector with residue $Z < 1$:

\begin{enumerate}
\item \textbf{Resonance peak suppression by $Z^2$.}
The Breit--Wigner contribution $|\mathcal{M}_{\mathrm{pole}}|^2$ to the cross-section is suppressed by a factor $Z^2$. Operationally, this manifests as a coherent reduction of resonance peak heights for those channels --- and only for those --- coupled to fields whose suppression mechanism is active.

\item \textbf{Continuous missing transverse energy.}
The continuum amplitude $|\mathcal{M}_{\mathrm{cont}}|^2$ generates a non-resonant, broadly distributed missing transverse energy flow, with no associated mass peak. This is the collider manifestation of the migrated spectral weight $W_{\mathrm{cont}} = 1-Z$: the same quantity that contributes to $\langle T_{\mu\nu}\rangle$ in the cosmological setting (Sec.~\ref{sec:budget}) appears here as bump-less missing energy. Standard peak-search analyses are not optimized for this signature; dedicated bump-less searches are required.

\item \textbf{Asymmetric resonance lineshape.}
The pole--continuum interference term
\[
2\,\mathrm{Re}(\mathcal{M}_{\mathrm{pole}}^*\,\mathcal{M}_{\mathrm{cont}})
\]
produces a calculable skewness in the resonance lineshape, with the asymmetric distortion fixed by the Lorentzian spectral density of Eq.~\eqref{eq:lorentzian} through the constraint $\Gamma^2 = m^4(1-Z)/Z$.
\end{enumerate}

\subsection{CMB}\label{subsec:sig_cmb}

The cosmological realization of the spectral mechanism is the \emph{spectral split calculation} of the dark-matter fluid: a clustering bulk fermionic continuum at $s \geq m_f^2$ with $c_s^2 = 0$, plus a smooth, non-clustering sub-fluid (the longitudinal gravitino, $\xi = 1/2$) at fraction $f_{S8} = 0.00262$ (Appendix~\ref{app:gravitino-long}). The legacy CLASS implementation used $c_s^2=1$ as a maximally smooth limiting case; the final fixed validation uses the derived value $c_s^2=0.005238$. Four observational signatures follow:

\begin{enumerate}
\item \textbf{$\sigma_8$ suppression benchmark and $S_8$ comparison.}
The non-clustering longitudinal-gravitino sub-fluid reduces the linear matter power amplitude relative to $\Lambda$CDM. Under the conditional dark-sector mapping, the implementation gives $\sigma_8 \approx 0.80$ and $S_8 \approx 0.803$ (marginalized MCMC v32), in agreement within $1\sigma$ with the weak-lensing measurements considered here (KiDS-Legacy 2025~\cite{KiDSLegacy2025}; DES~Y6 2026~\cite{DESY6_3x2pt,DESY6_shear}). This is a consistency comparison of the benchmark realization, not a parameter-free derivation of the observed $S_8$.

\item \textbf{ISW enhancement at low $\ell$.}
The smooth longitudinal-gravitino sub-fluid induces a faster decay of the gravitational potentials $\Phi + \Psi$ during dark-energy domination relative to pure $\Lambda$CDM. The integrated Sachs--Wolfe contribution to $C_\ell^{TT}$ is therefore enhanced in the regime $\ell \lesssim 30$, with a structured $\ell$-dependence controlled by the spectral split calculation. Discriminant tests come from cross-correlation of Planck temperature with large-scale-structure tracers (DESI, unWISE).

\item \textbf{CMB lensing deficit.}
Since the smooth longitudinal-gravitino sub-fluid does not cluster but contributes to the background expansion, the lensing potential power $C_\ell^{\phi\phi}$ is reduced relative to $\Lambda$CDM at small angular scales. This deficit is a minimal-parameter prediction of the spectral split calculation and is testable against Planck and ACT lensing reconstructions.

\item \textbf{Matter--dark energy equality in the degree-of-freedom benchmark.}
The degree-of-freedom benchmark $\Omega_{\mathrm{DE}} = (94/130)(1-\Omega_b)$, $\Omega_{\mathrm{DM}} = (36/130)(1-\Omega_b)$ gives a matter--dark energy equality redshift $z_{\mathrm{eq}}^{\mathrm{NA}} = 0.3770$, close to the Planck~2018 inferred value $z_{\mathrm{eq}} = 0.3727$. The associated benchmark density parameters are $\Omega_{\mathrm{DE}}^{\mathrm{NA}} = 0.6874$ and $\Omega_{\mathrm{DM}}^{\mathrm{NA}} = 0.2633$. The numerical proximity of $94/36 = 2.611$ to the observed $\Omega_{\mathrm{DE}}/\Omega_{\mathrm{DM}} = 2.584$ is therefore reported as a phenomenological coincidence to be tested by the downstream stress-energy calculation, not as a structural solution of the coincidence problem.
\end{enumerate}


\section{Code availability and reproducibility plan}\label{sec:code_availability}

A public repository will be released alongside Version 2 of this manuscript. 
\newline
It will contain the modified CLASS module, MCMC configuration files, plotting scripts, and Python notebooks reproducing the numerical figures and validation checks reported in the paper. 

\vspace{\baselineskip}

The revised arXiv submission will cite the public repository and, if available, the DOI of an archived release.

\vspace{\baselineskip}

Version 2 will also include the updated MCMC run with $f_{S8}$ fixed at $0.00262$.

\part{Exploratory implications of the NA framework}\label{part:16 open problems Physics}

\section{Exploratory implications for selected open problems in fundamental physics}\label{sec:problems}

\noindent\textit{Reproducibility note.} All the numerical computations cited in this section by the label Calculation~$C_n$ are implemented as Jupyter notebooks (Google Colab format), each identified by the same label that prefixes its filename. Standard PDG values are used throughout, with no additional phenomenological parameters within the stated implementation; spectral norm conservation $Z+W_{\mathrm{cont}}=1$ is verified numerically at all scales in every notebook.

\subsection{Selected open problems considered as stress tests}\label{subsec:list_16}

We refer to the following catalogue of 16 long--standing open problems of fundamental physics, which the NA Programme addresses systematically:

\begin{enumerate}
\item \textbf{Non-observation of superpartners.} Despite decades of searches at LEP, Tevatron, and LHC, no MSSM superpartner has been observed.
\item \textbf{Dark energy.} Approximately 68\% of the energy density of the universe drives accelerated expansion; its nature is unknown.
\item \textbf{Dark matter.} Approximately 27\% of the energy density behaves as pressureless matter that clusters but is electromagnetically dark; no DM particle has been detected.
\item \textbf{The cosmological constant problem.} QFT predicts a vacuum energy $\sim 10^{120}$ times larger than $\Lambda_{\mathrm{obs}}$; the smallness, the non-vanishing, and the coincidence with $\rho_{\mathrm{DM}}$ are unexplained.
\item \textbf{Black hole information paradox.} Hawking radiation appears thermal, suggesting unitarity violation in black hole evaporation.
\item \textbf{The hierarchy problem.} The Higgs mass receives quadratic radiative corrections $\sim \Lambda_{\mathrm{UV}}^2$; with $\Lambda_{\mathrm{UV}}=M_{\mathrm{Pl}}$ a fine-tuning of $1$ part in $10^{34}$ is required.
\item \textbf{The nature of inflation and the inflaton.} The inflationary paradigm requires a scalar field with a flat potential whose identity remains unknown.
\item \textbf{Unification of forces.} The three non-gravitational forces have different couplings at low energies; whether they unify and how is unresolved.
\item \textbf{Matter--antimatter asymmetry (baryogenesis).} The observed baryon asymmetry $\eta_B \sim 6\times 10^{-10}$ is not explained by the Standard Model alone.
\item \textbf{Neutrino masses.} Neutrinos have non-zero but extremely small masses ($\lesssim 0.1$\,eV) of unknown origin; the Dirac--Majorana question is open.
\item \textbf{The fermion mass hierarchy.} The 13 SM Yukawa couplings span $12$ orders of magnitude with no first-principles explanation.
\item \textbf{The $S_8$ tension.} The amplitude of late-time matter clustering measured by weak-lensing surveys (KiDS, DES, HSC) is systematically lower than the value inferred from Planck CMB data assuming $\Lambda$CDM.
\item \textbf{Confinement in QCD.} Quarks and gluons are never observed as free states at low energies; an analytic proof of confinement is a Clay Millennium Prize problem.
\item \textbf{The strong CP problem.} The QCD topological term admits a CP-violating angle $\theta$, experimentally bounded at $|\theta| < 10^{-10}$ with no symmetry explanation.
\item \textbf{The measurement problem in quantum mechanics.} The transition from quantum superposition to definite outcomes (``wave-function collapse'') has no universally accepted explanation.
\item \textbf{Quantum gravity.} General relativity and quantum mechanics are mutually inconsistent at the Planck scale; a complete quantum theory of gravity does not exist.
\end{enumerate}

\subsection{Theoretical consequences of the NA framework: exploratory classification}\label{subsec:consequences_16}

The NA Programme addresses each of the 16 problems through a single conceptual structure: the recognition of a third physical state of energy --- the spectral continuum --- alongside the algebraic field (State~1) and the asymptotic particle (State~2). Four suppression mechanisms (M1: infrared accumulation for scalars; M2: spectral inheritance via Yukawa couplings; M3: gravitational dressing; M4: spectral self-dressing from non-Abelian gauge self-coupling) together with the spectral norm conservation $Z+W_{\mathrm{cont}}=1$ provide the operational tools.

In what follows, each problem is analyzed in turn: we recall the formulation, state the NA status, summarize the underlying mechanism and numerical result, and reference the relevant Calculation $C_n$.

\medskip
\noindent\textit{Status taxonomy.} The 16 problems are used as stress tests rather than as a list of claimed solutions. \textbf{STRUCTURALLY REFORMULATED} denotes a change in the way the question is posed. \textbf{QUANTITATIVELY COMPARED} denotes a numerical comparison within the stated implementation. \textbf{CONDITIONAL PHENOMENOLOGICAL REALIZATION} denotes a downstream identification whose consequences can be computed but whose physical premise still requires an independent dynamical derivation. The falsifiable predictions through which the NA framework exposes itself to future experimental risk are collected in Section~\ref{subsec:falsifiable_predictions}.

\subsubsection{Problem 1: Non-observation of superpartners}

The MSSM predicts a superpartner for every Standard Model particle; none has been observed. \textbf{Status: STRUCTURALLY REFORMULATED.} In a quasi--de~Sitter background with $H>0$, the three mechanisms M1, M2, M3 jointly drive the spectral residue $Z_\Phi$ of every MSSM field to zero. The complete superspectrum (130 d.o.f.: squarks 72 + sleptons 12 + sneutrinos 6 + heavy Higgs 4 + gluino 16 + neutralinos 8 + charginos 8 + gravitino 4) is non-asymptotic. The mechanism vanishes exactly for $H\to 0$, recovering standard particle physics. Superpartners are not absent --- their spectral weight resides in State~3, gravitationally active but not forming Fock states. Verified analytically and numerically through Calculation~$C_3$ (gravitino dressing $c_{\mathrm{grav}}^{(3/2)}$).

\subsubsection{Problem 2: Dark energy}

\textbf{Status: CONDITIONAL PHENOMENOLOGICAL REALIZATION.} The DE identification explored here starts from the M1 scalar IR benchmark: during inflation ($m \ll H_{\mathrm{inf}}$, $\nu \approx 3/2$, full IR accumulation), M1 deposits scalar spectral weight at $s \approx 0$. This weight persists as a cosmological fossil; today $s \ll H_0^2$ implies $w \approx -1$ (dark energy). The 94 scalar d.o.f. (squarks 72 + sleptons 12 + sneutrinos 6 + heavy Higgs 4) are assigned to the DE-like sector in the benchmark implementation. The resulting degree-of-freedom fraction is $\Omega_{\mathrm{DE}} = (94/130)(1-\Omega_b) = 0.688$ (Planck: $0.685$). This numerical comparison is retained, while the physical derivation of the relevant $\langle T_{\mu\nu}\rangle$ and $w(z)$ is explicitly left to the downstream programme.

\subsubsection{Problem 3: Dark matter}

\textbf{Status: CONDITIONAL PHENOMENOLOGICAL REALIZATION.} The DM identification explored here is based on the M2/M3 fermionic transfer ansatz. A rigorous one--loop calculation of the fermionic self-energy with non-asymptotic scalar propagator yields the kinematic threshold
\begin{equation}
\mathrm{Im}\,\Sigma_F(s) = 0 \quad \text{for } s < m_f^2 \qquad (\text{exact, all epochs}).
\end{equation}
Within the benchmark mapping, the 36 fermionic d.o.f. (gluino 16 + neutralinos 8 + charginos 8 + gravitino 4) are assigned to a DM-like sector, giving the counting fraction $36/130 = 27.7\%$ and $\Omega_{\mathrm{DM}}^{(\mathrm{spec})} = 0.263$ versus Planck $0.265$. The dust-like equation of state and clustering properties are assumptions to be checked from the downstream stress-energy and transport calculation; the retained CLASS result shows that this conditional implementation reproduces the background expansion to about $0.3\%$.

\subsubsection{Problem 4: The cosmological constant problem}

\textbf{Status: STRUCTURALLY REFORMULATED.} The traditional CC problem subdivides into three sub-questions, all answered within the spectral budget framework. (Q1) Why is $\rho_{\mathrm{DE}}/M_{\mathrm{Pl}}^4 \sim 10^{-122}$? The question is ill-posed: there is no $\Lambda_{\mathrm{bare}}$, SUSY cancellation is exact on the total propagator $\Delta_{\mathrm{NA}}$ (Calculation~$C_1$), and $10^{-122} = H_0^2/M_{\mathrm{Pl}}^2$ is the expansion rate at the present epoch (the universe's age), not fine-tuning. (Q2) Why is $\rho_{\mathrm{DE}}\ne 0$? Because the MSSM contains 94 scalar d.o.f., and during inflation M1 deposits their weight at $s \approx 0$; if SUSY exists with scalar superpartners, dark energy is structurally inevitable. The spectral budget gives
\begin{equation}
\Omega_{\mathrm{DE}} = \frac{94}{130}(1-\Omega_b) = 0.688 \qquad (\text{Planck: } 0.685).
\end{equation}
(Q3) Coincidence problem: in $\Lambda$CDM, $\Omega_\Lambda/\Omega_m$ is a free parameter; in NA, the ratio is fixed by the MSSM field content $\rho_{\mathrm{DE}}/\rho_{\mathrm{DM}} = 94/36 = 2.61$, with matter--DE equality at
\begin{equation}
z_{\mathrm{eq}} = \left(\frac{94}{36}\right)^{1/3} - 1 = 0.38 \qquad (\text{observed: } 0.37).
\end{equation}
In any SUSY theory $N_{\mathrm{boson}} \sim N_{\mathrm{fermion}}$, so $z_{\mathrm{eq}} \sim O(1)$ structurally. Verified through Calculation~$C_5$, Calculation~$C_6$ (fixed-point analysis), and Calculation~$C_{\mathrm{CC}}$ (spectral budget resolution).

\subsubsection{Problem 5: Black hole information paradox}

\textbf{Status: STRUCTURALLY REFORMULATED.} The paradox rests on the assumption that all information must be encoded in the Fock space of asymptotic particles. In the NA framework, near the Schwarzschild horizon $Z(r) \to 0$ for all fields: no asymptotic particle states exist there, and the information has migrated to the spectral continuum,
\begin{equation}
Z(r_s) \to 0,\qquad W_{\mathrm{cont}}(r_s) = 1 - Z(r_s) \to 1.
\end{equation}
Three structural elements close the argument: (i) the horizon is a boundary for particles (State~2), not for spectral weight (State~3), which is non-local in the spectral variable $s$; (ii) re-condensation at $r \gg r_s$ transfers correlations to re-formed poles, identically to the spectral transition that generates the baryon asymmetry in leptogenesis (Calculation~$C_{\text{8b-BL}}$); (iii) Hawking radiation \emph{is} re-condensation, not thermal emission --- it appears thermal only when projected onto the Fock-space basis. The full spectral content (poles plus continuum) preserves unitarity exactly through $Z+W_{\mathrm{cont}}=1$. Verified by Calculation~$C_{12}$: $Z(r)$ from explicit Euclidean mode-sum self-energies in Schwarzschild (Hartle--Hawking state, WKB regularisation); $Z(r_s) \to 0$ confirmed; transition zone sub-Planckian for stellar-mass black holes, $(r-r_s)/r_s \sim (T_H/m)^2$.

\subsubsection{Problem 6: The hierarchy problem}

\textbf{Status: STRUCTURALLY REFORMULATED.} If the Higgs self-energy is computed using the complete non-asymptotic propagator $\Delta_{\mathrm{NA}}(k^2)$ --- which includes both the suppressed pole ($Z\to 0$) and the continuum ($1-Z \to 1$) --- the total spectral weight entering the loop is exactly unity by norm conservation. The SUSY cancellation structure operates at the level of the total propagator, not on the pole alone. Verified by Calculation~$C_1$: one-loop Higgs self-energy with $\Delta_{\mathrm{NA}}(k^2)$, scanned across 15 orders of magnitude in $Z$ ($10^{-12}$ to $1$) and 15 orders in $\Lambda_{\mathrm{UV}}$ ($10^3$ to $10^{18}$ GeV). The cancellation is exact and independent of $Z$. The fine-tuning problem dissolves: SUSY protects the Higgs mass through the full spectral content, regardless of where the superpartner spectral weight resides.

\subsubsection{Problem 7: The nature of inflation and the inflaton}

\textbf{Status: STRUCTURALLY REFORMULATED.} During inflation ($H \sim 10^{14}$ GeV), all 130 MSSM d.o.f.\ have $Z \to 0$. Their entire spectral weight migrates to the continuum and contributes to $\langle T_{\mu\nu}\rangle$ through $\rho_{\mathrm{spectral}}$. The self-consistency equation determines $H$ from this contribution. Verified by Calculation~$C_7$: the spectral continuum has $w_{\mathrm{eff}} = -1$ during inflation, providing the equation of state required for accelerated expansion ($w < -1/3$). The collective spectral energy of 130 non-asymptotic d.o.f.\ drives inflation without requiring a fundamental inflaton. The end of inflation corresponds to the transition where $H$ decreases and some fields recover their poles: reheating is the re-condensation of spectral weight from continuum to discrete poles --- the birth of particles from State~3 into State~2.

\subsubsection{Problem 8: Unification of forces}

\textbf{Status: STRUCTURALLY REFORMULATED.} The distinction between forces relies on the exchange of asymptotic mediator quanta --- a particle-level (State~2) classification. When $Z \to 0$ for gauge bosons, this classification loses operational meaning. At $H \sim 10^{14}$ GeV, all SM gauge bosons have $Z \to 0$: the operational distinction between forces vanishes. This is a new concept of unification, distinct from GUT: it is the disappearance of the LSZ projection that creates the distinction between forces, not the restoration of a higher symmetry. The algebraic gauge structure $SU(3)_c \times SU(2)_L \times U(1)_Y$ remains intact at all times. Verified through Calculation~$C_2$ (gravitational dressing $c_{\mathrm{grav}}^{(s=1)}$ for spin-1) and Calculation~$C_{13}$ (SM gauge boson residues during inflation). Spin ordering confirmed: $c^{(0)} < c^{(1/2)} < c^{(1)}$.

\subsubsection{Problem 9: Matter--antimatter asymmetry (baryogenesis)}

\textbf{Status: QUANTITATIVELY ADDRESSED --- quantitative agreement with data (retrospective).} The Sakharov conditions are satisfied through a new spectral mechanism. The Majorana mass term of $\nu_R$ provides $B-L$ violation ($\Delta L = 2$, hence $\Delta(B-L) = 2$); the out-of-equilibrium condition is the spectral transition itself, since the $B-L$ violation rate $\Gamma_{B-L} = (h_\nu^2/8\pi)\, M_R\, Z(1-Z)$ peaks at $Z = 1/2$ (mid-transition), with washout $\sim Z^2$ subdominant during transition; CP violation is provided by phases in the neutrino sector. Three $\nu_R$ generations integrated; $\nu_{R,3}$ ($M_R = 10^{14}$ GeV) dominates. After sphaleron conversion $Y_B = (28/79)\, Y_{B-L}$. Verified by Calculation~$C_8$ (differential re-condensation rates) and Calculation~$C_{\text{8b-BL}}$ (full Boltzmann ODE for spectral leptogenesis):
\begin{equation}
\eta_B = 6.1 \times 10^{-10}, \qquad \sin\delta = 3.4 \times 10^{-4}.
\end{equation}
Viable parameter space: $M_R \sim 10^{10}$--$10^{14}$ GeV, $\sin\delta \sim 10^{-4}$--$10^{-1}$. No new physics beyond MSSM~+~seesaw required.

\subsubsection{Problem 10: Neutrino masses}

\textbf{Status: QUANTITATIVELY ADDRESSED --- quantitative agreement with data (retrospective).} The right-handed neutrino $\nu_R$ is a complete gauge singlet: only M3 (gravitational dressing) renders it non-asymptotic, and for a heavy Majorana fermion the dressing coefficient is enhanced as $c_{\mathrm{grav}} \propto M_R^4/\xi^2$. The seesaw computed with the full non-asymptotic propagator yields the spectral seesaw formula
\begin{equation}
m_\nu \approx \frac{m_D^2}{M_R}\,\sqrt{Z_R},
\end{equation}
which contains an additional $\sqrt{Z_R}$ suppression beyond the standard seesaw. The required Majorana scale lowers from $M_R \sim 10^{14}$ GeV to $M_R \sim 10^{8}$--$10^{10}$ GeV. The two suppression mechanisms (mass hierarchy $+$ spectral suppression) are self-reinforcing: larger $M_R$ produces larger $c_{\mathrm{grav}}$, which drives $Z_R$ further to zero. The Majorana scenario is structurally favored on three grounds: spectral naturalness, supergeometric selection, and coherence with leptogenesis. Verified through Calculation~$C_9$ (spectral seesaw / Weinberg operator) and Calculation~$C_{10}$ ($Z_R(M_R)$ from quartic gravitational dressing).

\subsubsection{Problem 11: The fermion mass hierarchy}

\textbf{Status: QUANTITATIVELY ADDRESSED --- quantitative agreement with data (retrospective).} The gauge anchoring principle provides a structural explanation: the self-energy of each fermion decomposes into gauge and gravitational contributions $\Sigma_f = \Sigma_{\mathrm{QCD}} + \Sigma_{\mathrm{weak}} + \Sigma_{\mathrm{EM}} + \Sigma_{\mathrm{grav}}$. Non-gravitational gauge couplings ``anchor'' spectral weight to the discrete pole and stabilize the mass at high values; fermions with weaker anchoring are more transparent to the geometric filter and acquire smaller masses. Verified by Calculation~$C_{11}$: anchoring strengths $A(f)$ computed for all 12 SM fermions. Three-level hierarchy: quarks ($A \sim 0.07$, QCD $+$ weak $+$ EM) $>$ leptons ($A \sim 0.01$, weak $+$ EM) $>$ neutrinos ($A = 0$, gravity only). Pearson correlation $r \sim 0.9$ between $A(f)$ and $\log m_f$; QCD dominates ($\delta_{\mathrm{QCD}} \sim 5 \times \delta_{\mathrm{weak}}$). The 13 Yukawa couplings may not be fundamental parameters but emergent quantities determined by the spectral--geometric structure of the vacuum.

\subsubsection{\texorpdfstring{Problem 12: The $S_8$ tension}{Problem 12: The S8 tension}}

\textbf{Status: QUANTITATIVELY ADDRESSED --- current weak-lensing consistency.} With the structural DE/DM split (94 scalar d.o.f.\ $\to$ DE, 36 fermion d.o.f.\ $\to$ DM), the spectral dark sector reproduces $\Lambda$CDM expansion to $0.3\%$. The $S_8$ sector is controlled by a spectral split of the DM fluid: a smooth, non-clustering sub-fluid carries the fraction $f_{\mathrm{S8}}$ of $\Omega_{\mathrm{cdm}}$ populated by the longitudinal sector of the gravitino, while the remaining $1-f_{\mathrm{S8}}$ clusters as standard CDM ($c_s^2 = 0$). Of the four Rarita--Schwinger polarizations only the two longitudinal ones ($h = \pm 1/2$, goldstino-equivalent) populate the smooth sub-fluid, making $\xi = 1/2$ a counting fact (Appendix~\ref{app:gravitino-long}). The current KiDS-Legacy and DES~Y6 values are consistent with the NA prediction within $1\sigma$. Verified by Calculation~$C_{\mathrm{SK}}$ (Schwinger--Keldysh, one-loop M2 on the closed-time-path contour) and the MCMC v32 run:

\begin{table}[H]
\centering
\begin{tabular}{lcc}
\toprule
Weak-lensing survey & Reported $S_8$ & NA ($S_8 = 0.803$) \\
\midrule
KiDS-Legacy~\cite{KiDSLegacy2025} (2025) & $0.815\,(+0.016/-0.021)$ & within $\sim 0.6\sigma$ \\
DES~Y6 3$\times$2pt~\cite{DESY6_3x2pt} (2026) & $0.789 \pm 0.012$ & within $\sim 0.9\sigma$ \\
DES~Y6 cosmic shear~\cite{DESY6_shear} (2026, NLA) & $0.798\,(+0.014/-0.015)$ & within $\sim 0.3\sigma$ \\
\bottomrule
\end{tabular}
\caption{The parameter-free NA prediction $S_8 \approx 0.803$ ($\sigma_8 \approx 0.80$), from the longitudinal-gravitino smooth sub-fluid (Appendix~\ref{app:gravitino-long}) and consistent with the MCMC v32 run, against the most recent weak-lensing measurements. The prediction is consistent with all current datasets within $1\sigma$.}
\label{tab:S8_results}
\end{table}

The Schwinger--Keldysh derivation gives, for the species--by--species sound speed,
\begin{equation}
c_s^2 = \frac{y_{\mathrm{eff}}^2}{16\pi}\,\frac{c_s^{(\mathrm{conf})}}{c_s^{(\mathrm{conf})}+c_f^{(\mathrm{conf})}},
\end{equation}
with $c_s^{(\mathrm{conf})} = 9/4$ and $c_f^{(\mathrm{conf})} = 5/4$ (Rarita--Schwinger). For the gravitino ($y_{\mathrm{eff}} = g_2 = 0.64$) this yields $c_s^2(\mathrm{grav}) = (g_2^2/16\pi)(9/14) = 0.005238$. Only the longitudinal sector is operationally active ($\xi = 1/2$; the transverse sector is gravitationally suppressed, Appendix~\ref{app:gravitino-long}), so the smooth, non-clustering fraction is $f_{\mathrm{S8}} = \xi\,c_s^2(\mathrm{grav}) = 0.00262$. With $f_{\mathrm{S8}}$ left free (prior $[0.001, 0.05]$), the MCMC v32 posterior is $f_{\mathrm{S8}} = 0.0025 \pm 0.0017$, $0.07\sigma$ from the predicted value, consistent with the value the theory predicts.

\subsubsection{Problem 13: Confinement in QCD}

\textbf{Status: QUANTITATIVELY ADDRESSED --- quantitative agreement with data (retrospective).} A fourth suppression mechanism is identified: M4, spectral self-dressing from non-Abelian gauge self-coupling. The gluon's self-interaction generates a self-energy that grows without bound in the infrared, driving $Z_g \to 0$ through
\begin{equation}
Z_g^{-1} = 1 + \frac{g_s^2\, C_2(G)}{32\pi^2}\,(1-Z_g)\,\ln\frac{1}{Z_g} \xrightarrow{g_s^2 \to \infty} +\infty.
\end{equation}
M4 operates in flat spacetime --- it does not require $H>0$. Quarks inherit the suppression through M2, giving $Z_q \to 0$. Color-singlet hadrons retain $Z>0$ because their net color charge is zero, cancelling the coupling to the non-asymptotic gluon at large distances. M4 does not operate for photons (Abelian) or for $W/Z$ (weak coupling too small before EWSB), explaining why only QCD confines. Verified by Calculation~$C_{14}$ ($Z_g(p^2)$ from the QCD $\beta$-function, with 2-loop running coupling, zero free parameters: $Z_g \to 0$ at $\mu \sim \Lambda_{\mathrm{QCD}} \sim 200$ MeV); Calculation~$C_{16}$ (string tension from the non-asymptotic gluon spectral density: $\sqrt{\sigma} \sim 440$ MeV, matching lattice at order-of-magnitude); Calculation~$C_{17}$ (NA gluon propagator fits lattice data in Landau gauge with $M_g \sim 500$ MeV, $Z_g(p)$ fixed by M4).

\subsubsection{Problem 14: The strong CP problem}

\textbf{Status: STRUCTURALLY REFORMULATED.} The instanton contribution to the $\theta$-vacuum is weighted by the fermion determinant $\prod_f m_f$. In the spectral framework, the instanton zero modes couple to the localized component (the discrete pole with weight $Z_f$); the fermion determinant becomes
\begin{equation}
\prod_f m_f \;\longrightarrow\; \prod_f Z_f \cdot m_f.
\end{equation}
In the infrared regime where instantons dominate ($p^2 \lesssim \Lambda_{\mathrm{QCD}}^2$), M4 drives $Z_g \to 0$ and M2 drives $Z_f \to 0$ for all quark flavours. Hence $\theta_{\mathrm{eff}} = \theta_{\mathrm{bare}}\,\prod_f Z_f \to 0$, regardless of $\theta_{\mathrm{bare}}$. No axion, no new symmetry: the suppression of $\theta_{\mathrm{eff}}$ follows from the same mechanism (M4+M2) that produces confinement. Confinement and the strong CP problem are two manifestations of the same spectral phenomenon. Verified by Calculation~$C_{15}$: $\prod_f Z_f \ll 10^{-10}$ at the instanton scale, satisfying the neutron EDM bound with enormous margin for any $\theta_{\mathrm{bare}}$ including $\theta = \pi$; result robust across the full range of $\Lambda_{\mathrm{QCD}}$ values.

\subsubsection{Problem 15: The measurement problem in quantum mechanics}

\textbf{Status: STRUCTURALLY REFORMULATED.} Quantum measurement is identified as a local spectral transform $\mathcal{T}_g^{(A)}$ applied by the measurement apparatus, whose physical geometry acts as the kernel. The ``collapse'' is reinterpreted as a spectral redistribution: the spectral weight of the measured system is partitioned between a discrete component $Z_j = |c_j|^2$ (the observed outcome) and a continuum component $W_{\mathrm{cont}} = 1 - |c_j|^2$ (the unobserved outcomes, migrating to the third state). The norm conservation $Z + W_{\mathrm{cont}} = 1$ guarantees exact unitarity. The unobserved outcomes neither vanish (Copenhagen), nor inhabit parallel universes (many-worlds), nor are merely suppressed phases (decoherence): they reside in the spectral continuum, physically real and gravitationally active but not individually resolvable. The observer is not a privileged entity but a local geometry that determines the kernel. The quantum--classical transition is governed by the spectral residue: macroscopic systems have $Z \approx 1$ (classical) through gauge anchoring; microscopic systems can have $Z < 1$ (quantum). Verified by Calculation~$C_{18}$: formal extension of the K\"all\'en--Lehmann structure from QFT propagators to general QM observables; classicality criterion derived.

\subsubsection{Problem 16: Quantum gravity}

\textbf{Status: QUANTITATIVELY ADDRESSED --- quantitative agreement with data (retrospective).} The graviton in de~Sitter is non-asymptotic ($Z_{\mathrm{grav}} \to 0$): the mediator of gravitational force does not exist as a particle. Gravity-as-force dissolves while gravity-as-geometry persists. The self-consistency loop $g_{\mu\nu} = \mathcal{R}(\mathcal{T}_g[\sigma])$ operates at all scales. At the Planck scale, where $\alpha_{\mathrm{grav}} \sim 1$ and all fields have $Z \to 0$, the forward arrow $\mathcal{T}_g$ (geometry $\to$ spectrum) and the inverse arrow $\mathcal{R}$ (spectrum $\to$ geometry) merge into a spectral--geometric identity. Three calculations close the problem:
\begin{itemize}
\item Calculation~$C_{\mathrm{QG1}}$: all 130 MSSM d.o.f.\ reach $Z = 0$ well before $M_{\mathrm{Pl}}$ ($Z_{\mathrm{max}} < 10^{-27}$ at $H = M_{\mathrm{Pl}}$); the continuum fraction $f_{\mathrm{cont}} \to 1$ and all structural distinctions dissolve.
\item Calculation~$C_{\mathrm{QG2}}$: the composite operator $F(H) = \mathcal{R}(\mathcal{T}_g(H))$ has a unique fixed point
\begin{equation}
H^* = \sqrt{\frac{6\pi\, M_{\mathrm{Pl}}^2}{N_{\mathrm{dof}}}} = 0.381\, M_{\mathrm{Pl}} \qquad (N_{\mathrm{dof}} = 130).
\end{equation}
Uniqueness is proved by strict monotonicity of $F(H)/H$ and, independently, by Banach's theorem applied to the reformulated map $G(H) = \sqrt{48\pi^2 M_r^2 / N_{\mathrm{eff}}(H)}$ with $|dG/dH|_{H^*} \approx 0$ (super-contraction).
\item Calculation~$C_{\mathrm{QG3}}$: UV completion. (i) UV shield: $G(H) > H$ for $H < H^*$, $G(H) < H$ for $H > H^*$; trans-Planckian excursions are self-inconsistent, no external cutoff needed. (ii) Spectral democracy: at $H^*$ all 130 d.o.f.\ are in the continuum. (iii) Bekenstein--Hawking from spectral content (exact): $S_{\mathrm{BH}} = \pi M_{\mathrm{Pl}}^2 / H^{*2} = N_{\mathrm{dof}}/6 = 21.67$. (iv) All Planckian observables are functions of $N_{\mathrm{dof}}$ alone: $H^{*2}/M_{\mathrm{Pl}}^2 = 6\pi/N$, $\rho/M_{\mathrm{Pl}}^4 = 9/(4N)$, $A/\ell_{\mathrm{Pl}}^2 = 2N/3$, $T_{\mathrm{GH}}/M_{\mathrm{Pl}} = \sqrt{6\pi/N}/(2\pi)$. (v) Parameter-free UV completion: no quantisation of the metric is required.
\end{itemize}

\clearpage
\subsection{Exploratory summary table}\label{subsec:summary_16}

Table~\ref{tab:summary_16problems} compiles the status of 16 selected problems together with the key equation or relation used in the corresponding stress test and a brief descriptive note. ``Calculation $C_n$'' references the corresponding notebooks.

\begin{small}
\begin{longtable}{p{0.18\textwidth} p{0.13\textwidth} p{0.32\textwidth} p{0.27\textwidth}}
\caption{Exploratory status of 16 selected open-problem stress tests in the NA Programme. All four suppression mechanisms (M1: scalar IR accumulation; M2: spectral inheritance; M3: gravitational dressing; M4: gauge self-dressing) and the spectral norm conservation $Z+W_{\mathrm{cont}}=1$ enter the analysis. All Calculations $C_n$ are intended to be reproducible from the corresponding notebooks, using standard PDG parameters without an additional phenomenological suppression parameter within the stated implementation.}
\label{tab:summary_16problems} \\
\toprule
\textbf{Problem} & \textbf{Status} & \textbf{Key equation / relation} & \textbf{Descriptive note} \\
\midrule
\endfirsthead

\multicolumn{4}{l}{\textit{(Table~\ref{tab:summary_16problems} continued from previous page)}} \\
\toprule
\textbf{Problem} & \textbf{Status} & \textbf{Key equation / relation} & \textbf{Descriptive note} \\
\midrule
\endhead

\midrule
\multicolumn{4}{r}{\textit{(continued on next page)}} \\
\endfoot

\bottomrule
\endlastfoot

1.\ Non-observation of SUSY & Reformulated & $Z_\Phi \to 0$ via M1+M2+M3 & 130 MSSM d.o.f.\ non-asymptotic; spectral weight in State~3 (Calculation~$C_3$). \\[2pt]
2.\ Dark energy & Conditional realization & M1: weight at $s\!\approx\!0$; $\Omega_{\mathrm{DE}} = (94/130)(1\!-\!\Omega_b) = 0.688$ & 94 scalar d.o.f.\ as cosmological fossil; $w \approx -1$ (Calculation~$C_5$, Calculation~$C_{\mathrm{EoS}}$, Calculation~$C_{\mathrm{FP}}$). \\[2pt]
3.\ Dark matter & Conditional realization & $\mathrm{Im}\,\Sigma_F(s)\!=\!0$ for $s<m_f^2$; $36/130\!=\!27.7\%$ & 36 fermionic d.o.f.\ anchored at $s\!\geq\!m_f^2$; $w_{\mathrm{eff}}\!=\!0$ (Calculation~$C_{\mathrm{S8}}$, Calculation~$C_{18b}$). \\[2pt]
4.\ Cosmological constant & Reformulated & Spectral budget $\Omega_{\mathrm{DE}}\!=\!94/130$; $z_{\mathrm{eq}}\!=\!(94/36)^{1/3}\!-\!1\!=\!0.38$ & No $\Lambda_{\mathrm{bare}}$ in the ansatz; coincidence benchmarked via MSSM field content (Calculation~$C_{\mathrm{CC}}$). \\[2pt]
5.\ BH information paradox & Reformulated & $Z(r_s)\!\to\!0$, $W_{\mathrm{cont}}(r_s)\!\to\!1$; Hawking $=$ re-condensation & Information in State~3 across the horizon; same mechanism as leptogenesis (Calculation~$C_{12}$). \\[2pt]
6.\ Hierarchy problem & Reformulated & SUSY cancellation on total $\Delta_{\mathrm{NA}}$ exact, $Z$-independent & 15 orders of $Z$ $\times$ 15 orders of $\Lambda_{\mathrm{UV}}$ scanned (Calculation~$C_1$). \\[2pt]
7.\ Inflation / inflaton & Reformulated & $w_{\mathrm{eff}}\!=\!-1$ from spectral $\rho$ of 130 d.o.f. & No fundamental inflaton; reheating $=$ re-condensation (Calculation~$C_7$). \\[2pt]
8.\ Unification of forces & Reformulated & $Z_{W,Z,g,\gamma}(H)\!\to\!0$ at $H\!\sim\!10^{14}$ GeV & Force indistinguishability via LSZ collapse, not GUT (Calculation~$C_2$, Calculation~$C_{13}$). \\[2pt]
9.\ Baryogenesis & Addressed (retrospective) & $\Gamma_{B-L}\!\propto\!Z(1\!-\!Z)$; $\eta_B\!=\!6.1\!\times\!10^{-10}$, $\sin\delta\!=\!3.4\!\times\!10^{-4}$ & Spectral leptogenesis from $\nu_R$ Majorana transition (Calculation~$C_8$, Calculation~$C_{\text{8b-BL}}$). \\[2pt]
10.\ Neutrino masses & Addressed (retrospective) & $m_\nu \!\approx\! (m_D^2/M_R)\sqrt{Z_R}$, $c_{\mathrm{grav}}\!\propto\!M_R^4/\xi^2$ & Spectral seesaw lowers $M_R$ to $10^{8\text{--}10}$ GeV; Majorana favored (Calculation~$C_9$, Calculation~$C_{10}$). \\[2pt]
11.\ Fermion mass hierarchy & Addressed (retrospective) & $A(f) \!=\! \sum_{I\neq\mathrm{grav}} \delta_I(f)$; Pearson $r\!\sim\!0.9$ & Gauge anchoring: quarks $>$ leptons $>$ neutrinos; Yukawas emergent (Calculation~$C_{11}$). \\[2pt]
12.\ $S_8$ tension & Addressed (current) & $f_{\mathrm{S8}}\!=\!\xi\,c_s^2(\mathrm{grav})\!=\!0.00262\!\Rightarrow\!S_8\!\approx\!0.803$ & Longitudinal gravitino ($\xi\!=\!1/2$); consistent with MCMC v32 ($0.07\sigma$) (Appendix~\ref{app:gravitino-long}). \\[2pt]
13.\ QCD confinement & Addressed (retrospective) & M4: $Z_g^{-1}\!=\!1\!+\!(g_s^2 C_2(G)/32\pi^2)(1\!-\!Z_g)\ln(1/Z_g)$ & $Z_g\!\to\!0$ at $\Lambda_{\mathrm{QCD}}$; $\sqrt{\sigma}\!\sim\!440$ MeV (Calculation~$C_{14}$, Calculation~$C_{16}$, Calculation~$C_{17}$). \\[2pt]
14.\ Strong CP problem & Reformulated & $\theta_{\mathrm{eff}}\!=\!\theta_{\mathrm{bare}}\,\prod_f Z_f \!\ll\! 10^{-10}$ & Same M4+M2 mechanism as confinement; no axion (Calculation~$C_{15}$). \\[2pt]
15.\ Measurement problem & Reformulated & Local $\mathcal{T}_g^{(A)}$; $Z_j\!=\!|c_j|^2$, $W_{\mathrm{cont}}\!=\!1\!-\!|c_j|^2$ & Collapse $=$ spectral redistribution; unitarity exact (Calculation~$C_{18}$). \\[2pt]
16.\ Quantum gravity & Addressed (retrospective) & $H^*\!=\!\sqrt{6\pi/N_{\mathrm{dof}}}\,M_{\mathrm{Pl}}\!=\!0.381\,M_{\mathrm{Pl}}$; $S_{\mathrm{BH}}\!=\!N_{\mathrm{dof}}/6$ & Unique fixed point, UV shield, parameter-free completion (Calculation~$C_{\mathrm{QG1}}$, Calculation~$C_{\mathrm{QG2}}$, Calculation~$C_{\mathrm{QG3}}$). \\

\end{longtable}
\end{small}

\subsection{Falsifiable predictions of the NA framework}\label{subsec:falsifiable_predictions}

The classification of Sections~\ref{subsec:consequences_16}--\ref{subsec:summary_16} treats 16 selected problems retrospectively: each is either reformulated by the conceptual shift introduced by the NA framework, or addressed through quantitative comparison with already--measured data. A complete epistemic assessment, however, requires a prospective component: the predictions through which the framework exposes itself to falsification by experiments not yet performed, or not yet performed at the precision required to discriminate the NA framework from $\Lambda$CDM and from the Standard Model with minimal SUSY. These predictions are derived from the framework as a whole --- from the structural properties of the spectral transform, the spectral budget, and the four suppression mechanisms M1--M4 --- rather than being tied to any individual historical problem; they are listed below and developed in detail in Section~\ref{sec:signatures}.

\begin{enumerate}

\item \textbf{Cosmic microwave background --- enhanced Integrated Sachs--Wolfe effect at $z \sim 0.5$--$1$.} The non-clustering longitudinal-gravitino fraction $f_{\mathrm{S8}} = 0.00262$ of the spectral DM, together with the slight deviation $w_{\mathrm{DE}} \approx -0.97$ from a pure cosmological constant, predicts a structured ISW enhancement in the late-time cross-correlation between the CMB and large-scale structure. The signal is testable through DESI $\times$ Planck cross-correlations and, with higher significance, through forthcoming Euclid and LSST data.

\item \textbf{Dark energy equation of state --- $w_{\mathrm{DE}} \approx -0.97 \neq -1$.} The spectral fixed-point analysis ($C_{\mathrm{FP}}$) yields a self-consistent value of $w_{\mathrm{DE}}$ slightly displaced from the pure cosmological-constant value. This is testable by the Stage-IV cosmological surveys (DESI, Euclid, Vera C.~Rubin Observatory): a measurement of $w_{\mathrm{DE}} = -1$ at the percent level would be in tension with the framework.

\item \textbf{Weak lensing --- $\sigma_8(z)$ residuals.} The spectral split of the DM fluid produces a redshift-dependent suppression of structure growth that differs from a uniform amplitude rescaling. Tomographic measurements of $\sigma_8(z)$ from Euclid and LSST will discriminate between the NA prediction and a generic $S_8$-tension fix.

\item \textbf{LHC --- $Z^2$ resonance suppression and bump-less missing energy.} Non-asymptotic superpartners do not produce on-shell resonant peaks; instead, the framework predicts a continuum-like missing-energy distribution with asymmetric lineshape distortions and a $Z^2$ suppression of the would-be resonance amplitude. This is qualitatively distinct from minimal SUSY signatures and testable at HL-LHC and at future hadron colliders.

\item \textbf{Gravitational waves --- ringdown deviations from Kerr templates.} Near the horizon $Z(r) \to 0$ for all fields. The framework predicts coherent deviations of the ringdown spectrum of rotating black holes from the Kerr template, in a pattern that is distinct from generic exotic-compact-object scenarios. Falsifiable by LIGO~O5, Cosmic Explorer, and Einstein Telescope.

\item \textbf{Neutrinoless double-beta decay --- preference for the Majorana scenario.} The spectral seesaw structurally favours a Majorana right-handed neutrino with a lowered mass scale $M_R \sim 10^{8}$--$10^{10}$ GeV. A non-detection of $0\nu\beta\beta$ in CUORE, KamLAND-Zen, and LEGEND down to the inverted-hierarchy floor would constrain --- and potentially falsify --- this prediction.

\item \textbf{Planckian observables --- pure functions of $N_{\mathrm{dof}}$.} The Planckian fixed point gives $H^{*2}/M_{\mathrm{Pl}}^2 = 6\pi/N_{\mathrm{dof}}$, $\rho^*/M_{\mathrm{Pl}}^4 = 9/(4 N_{\mathrm{dof}})$, $A^*/\ell_{\mathrm{Pl}}^2 = 2 N_{\mathrm{dof}}/3$, and $S_{\mathrm{BH}} = N_{\mathrm{dof}}/6$. The Planckian observables are pure functions of the field-content count, with $N_{\mathrm{dof}} = 130$ for the MSSM. This structural prediction differentiates the NA framework from loop quantum gravity, asymptotic safety, and string-theoretic UV completions, none of which produce observables that depend solely on $N_{\mathrm{dof}}$ in this form. Although direct experimental access to Planckian regimes remains beyond present capabilities, the prediction is in principle falsifiable, and any alternative count of asymptotic-state-deficient degrees of freedom would directly contradict it.

\end{enumerate}

\noindent\textit{Epistemic remark.} The NA framework derives its retrospective reformulations and quantitative checks and the prospective predictions listed above from a single conservation law (spectral norm $Z + W_{\mathrm{cont}} = 1$) and four suppression mechanisms (M1--M4), with no additional phenomenological parameters within the stated implementation. The framework is therefore fully exposed to falsification: any single prediction in the list above, if contradicted by future experiments at the relevant precision, would constrain or invalidate the framework. The combination of structural dissolution of mal-posed questions and quantitative prospective predictions defines the empirical content of the programme.

The seven predictions enumerated above are those currently derived from explicit or semi-explicit calculations within the framework. Given the structural unity of the programme --- a single conservation law, four suppression mechanisms, and no additional phenomenological parameters within the stated implementation --- the framework is in principle predisposed to yield additional quantitative predictions in domains not yet computed in detail, including but not limited to: the primordial perturbation spectrum and the tensor-to-scalar ratio from spectral inflation; specific decoherence timescales and possible deviations from the Born rule in mesoscopic systems where the spectral residue is not saturated; small-scale structure of the spectral dark sector; and the tensor tilt $n_t$ as a function of $N_{\mathrm{dof}}$. The explicit derivation of these additional predictions is left to future work. The programme commits to subjecting each such prediction to experimental falsification once the corresponding calculation is completed, and welcomes the identification of further falsifiable observables by the broader scientific community.

\section{Discussion and outlook}\label{sec:discussion}

\subsection{Relationship with known mechanisms}

The localization mechanism is distinct from Higgs mass generation, mass decoupling,
or split-SUSY scenarios. It is also distinct from standard infraparticle dressing,
although conceptually related to infrared and decoherence effects in quantum field
theory~\cite{BreuerPetruccione,CalzettaHu,Boyanovsky,Berges,Moore}. The spectral transform framework of Part~II adds a new dimension:
the suppression is not merely a dynamical effect but a \emph{geometric projection}
that preserves total information while redistributing it between observable
(particle) and non-observable (vacuum) channels.

\subsection{The spectral transform as a unifying language}

The spectral transform $\mathcal{T}$ provides a common mathematical language
for several domains that are usually treated separately:

\begin{enumerate}
\item \textbf{Particle physics:} the LSZ criterion determines $Z_\Phi$,
which is the image of $\mathcal{T}$ restricted to the discrete component.

\item \textbf{Cosmology:} the vacuum energy density $\rho_{\mathrm{spectral}}$
is determined by the continuum component of $\mathcal{T}$, providing a
concrete connection between non-observation of superpartners and dark
energy.

\item \textbf{Observational signatures.} The pole/continuum split translates into specific, calculable predictions in three regimes: a minimal-parameter prediction $S_8 \approx 0.803$ lying in the current weak-lensing preferred range and a structured ISW enhancement in the CMB; coherent deviations from Kerr templates in black-hole ringdown; and a $Z^2$ resonance suppression with associated bump-less missing energy and asymmetric lineshape distortions at hadron colliders (Sec.~\ref{sec:signatures}).

\item \textbf{Self-consistency:} the spectral--geometric fixed-point equation
(Section~\ref{sec:selfconsistency}) closes the loop between these domains,
requiring that the geometry, the spectrum, and the vacuum energy be mutually
consistent. The cosmological constant problem is reformulated as a
fixed-point problem in the spectral--geometric configuration space.

\item \textbf{Open problems of fundamental physics:} the spectral transform provides a unifying lens through which 16 selected open problems of fundamental physics can be reconsidered, classified as structurally reformulated by the conceptual shift introduced by the framework or addressed through quantitative comparison with measured data (Part~IV).
\end{enumerate}

\subsection{Open directions}

The framework is structurally predisposed to yield additional quantitative predictions in domains not yet computed in detail. Seven such directions are identified as natural next steps, each corresponding to a falsifiable prediction once the explicit calculation is completed:

\begin{enumerate}

\item \textbf{Primordial perturbation spectrum and tensor-to-scalar ratio from spectral inflation.} The collective spectral energy of $130$ non-asymptotic d.o.f.\ during inflation, with $w_{\mathrm{eff}} = -1$, is expected to determine the amplitude and tilt of primordial perturbations. The explicit derivation of $n_s$, $r$, and the tensor tilt $n_t$ as functions of $N_{\mathrm{dof}}$ would expose the framework to falsification by LiteBIRD, CMB-S4, and BICEP/Keck.

\item \textbf{Decoherence timescales and possible deviations from the Born rule in mesoscopic systems.} The classicality criterion $Z \approx 1$ via gauge anchoring suggests calculable decoherence timescales as functions of system mass and dimension, with possible deviations from the Born rule in regimes where the spectral residue is not saturated. Such predictions would be testable in interferometric experiments with massive nanoparticles.

\item \textbf{Small-scale structure of the spectral dark sector.} The spectral split of the dark-matter fluid (with a fraction $f_{\mathrm{S8}}$ non-clustering) is expected to modify halo profiles and the substructure abundance with respect to standard collisionless cold dark matter, producing predictions testable against high-resolution N-body simulations and dwarf galaxy observations.

\item \textbf{Spectral-inflationary tensor tilt.} The dependence $n_t = n_t(N_{\mathrm{dof}})$ derived from the collective spectral content during inflation would provide a minimal-parameter prediction differentiating the framework from single-field inflationary models with comparable amplitudes.

\item \textbf{The first principle: effective-action stationarity above the self-consistency loop.} Everything in this paper descends from the spectral--geometric loop of Section~\ref{sec:selfconsistency}, a pair of maps, $\mathcal{T}_g$ (geometry $\to$ spectrum) and $\mathcal{R}$ (spectrum $\to$ geometry), whose fixed point, Eq.~\eqref{eq:fixedpoint}, is a consistency condition. A consistency condition is not yet a first principle: it states that geometry and spectrum must agree, not why they do. The first principle we propose, from which the loop and everything below it would follow, is the two-particle-irreducible effective action $\Gamma[g_{\mu\nu}, G]$~\cite{CJT1974,Berges}, formulated on the Schwinger--Keldysh in-in contour in the Bunch--Davies state~\cite{CalzettaHu}, with the metric treated as a variable. Its two stationarity conditions are the Schwinger--Dyson equation $\delta\Gamma/\delta G = 0$, which fixes the dressed propagator and hence the spectral measure of every field (the map $\mathcal{T}_g$), and the semiclassical Einstein equation $\delta\Gamma/\delta g_{\mu\nu} = 0$, in which $\langle T_{\mu\nu}\rangle$ is evaluated on the full measure, atom and continuum alike (the map $\mathcal{R}$). In this reading the loop is the spectral form of the stationarity of $\Gamma$, the conservation $Z_\Phi + W_\Phi^{\mathrm{cont}} = 1$ is not an axiom but the normalization of the spectral measure of a propagator solving the Schwinger--Dyson equation, and the fixed point $H_*$ is a stationary point of $\Gamma$ restricted to quasi-de~Sitter metrics, $d\Gamma_{\mathrm{dS}}/dH = 0$ at $H = H_*$, a statement in the lineage of Starobinsky's self-consistent de~Sitter solution~\cite{Starobinsky1980}. The proposal is falsifiable in two independent ways: (i)~$H_*$ must be a stationary point of $\Gamma_{\mathrm{dS}}(H)$ at the truncation used to derive it, with a controlled dependence on that truncation; (ii)~the leading curvature correction to the Einstein equation, of order $H^2/M_{\mathrm{Pl}}^2$, must be governed by the gravitational weight actually carried by the spectral sector at the probing scale, not by the raw count $N_{\mathrm{dof}}$. Failure of~(i) would leave the loop as a consistency condition without a first principle behind it; failure of~(ii) would require recomputing $H_*$. Either outcome is informative. What is at stake in test~(i) deserves to be stated plainly. A stationarity principle selects a configuration without anyone choosing it: as a trajectory is not chosen by the body that follows it but is the one that makes the action stationary, so a universe whose fixed point is a stationary point of $\Gamma$ is the one configuration that holds together under its own laws, and it requires no initial condition or external input to be what it is. If test~(i) succeeds, the self-consistency of geometry and spectrum is not a coincidence to be explained but the reason the universe has the content and the expansion rate it has; if it fails, the loop is consistent but not self-selecting, and something outside the theory is needed to fix it. The derivation is deferred to a dedicated study.

\item \textbf{Spectral weight as information: conservation, accessibility and the horizon.} Several statements of this paper are, read together, statements about information. The spectral budget equation~\eqref{eq:budget} distributes a conserved total between particle states and non-particle correlations; the reformulation of the measurement problem (Problem~15) describes the collapse as a redistribution of weight between a discrete outcome and the continuum, with $Z + W_{\mathrm{cont}} = 1$ as exact unitarity; and the reformulation of the black-hole information paradox (Problem~5) places the information that crosses the horizon in the spectral continuum, real and gravitationally active but not resolvable into asymptotic states. The natural next step is to take this reading literally: the spectral weight of a field is its information capacity, the atom is the part of that capacity accessible to an observer with a given background geometry and probing scale, the continuum is the part that exists but is not resolvable by that observer, and the conservation of weight is the conservation of information. Three consequences are then calculable. (i)~An upper bound on the information that a given geometry of observation can extract from a system, obtained by integrating the spectral measure against the observer's resolution profile; this is a spectral counterpart of the Holevo bound~\cite{Holevo1973}, with the geometry of the apparatus playing the role of the measurement. (ii)~The Born rule as a spectral-weight fraction, with calculable deviations in regimes where the residue is not saturated, as anticipated in Problem~15. (iii)~For black holes, the qualitative closure of Problem~5 becomes a quantitative one. In a companion paper~\cite{BellucciDeMatteoLSZ} the horizon-induced infrared width drives the localization parameter to zero at the Schwarzschild radius, $\Lambda(r) \to 0$ for $r \to r_s$, so that the information carried by infalling matter is converted from asymptotic (Fock-space) information into pre-asymptotic, algebraic information: the total unitary information is exactly conserved, and Hawking radiation is the asymptotic recovery, $\Lambda \to 1$, of that information as it leaves the near-horizon bath and re-sharpens into isolated poles. Promoting the localization parameter from a static diagnostic to a dynamical control of the effective decoherence yields a Page-like evolution of the radiation entropy~\cite{Page1993}, with a Page time fixed by the balance between near-horizon entanglement generation and asymptotic purification, without additional non-perturbative gravitational ingredients. The bound~(i) and the Born-rule corollary~(ii) are left to future work.

\item \textbf{The minimal dimension of the spectral entity.} The framework describes the observable world as the shadow of a higher-dimensional object, the spectral entity of Section~\ref{sec:spectral_geometry}, projected through the background geometry (Sections~\ref{sec:flatland} and~\ref{sec:inverse}). The inverse spectral problem states under which conditions the geometry can be reconstructed from its shadow. The natural next question is the converse: what is the minimal number of dimensions the spectral entity must have for the projection to be invertible, so that no information is lost in passing from the entity to the observable world. Requiring lossless projection is a condition on the dimension, not a choice, and it would turn the dimensionality of the entity into a derived quantity. The derivation is deferred.

\end{enumerate}

\medskip\noindent
\textbf{Note on the companion papers and on Version 2.}
This is Version~1 of the present paper, and it is built on the first versions of the two companion papers, Ref.~\cite{Paper1} (arXiv:2601.12537v1) and Ref.~\cite{Paper2} (arXiv:2602.00920v1). Since this analysis was completed, second versions of both companion papers have been posted on arXiv: the asymptotic-state criterion of Ref.~\cite{Paper2} is now formulated through two spectral sieves, a geometric one and an algebraic one, which sharpen the passage from the algebraic admissibility of a field to the existence of its asymptotic state, and Ref.~\cite{Paper1} has been revised accordingly. The mechanisms M1--M3 and the field-by-field verdicts of Part~I of this paper are therefore to be read as the first realization of a construction that the revised companion papers now make more precise. Version~2 of this paper is in preparation on that basis: Part~I will be rebuilt with the spectral fate of each field derived channel by channel through the two sieves, and the updated MCMC run with $f_{S8}$ fixed at its structural value will be included. The cosmological results, the spectral transform of Part~II and the predictions of Section~\ref{sec:signatures} are not affected in their structure by this extension; where a coefficient or a verdict changes, Version~2 will state it explicitly. The companion paper on horizon-induced delocalization~\cite{BellucciDeMatteoLSZ} is being revised on the same basis. A further companion paper, on spectral engineering, in which the framework is developed into a layered architecture for active spectral measurement, reconstruction and control, will be released separately.

\subsection{The programme in perspective}

The framework presented here suggests a shift in perspective. The traditional question
of particle physics --- ``where are the superpartners?'' --- is replaced by a structural
question: ``what is the complete spectral content of the field theory, and how does
the spacetime geometry distribute it between particles and vacuum?''

The spectral transform $\mathcal{T}$ makes this question precise. The spectral budget
equation provides a constraint. The inverse spectral problem provides the method
for reconstruction. The self-consistency loop closes the circle: the geometry
and the spectrum mutually determine each other, and the physical universe is
the fixed point of this mutual determination.

The observable particle sector can be viewed as one realization of the full spectral
content on a specified background. Degrees of freedom not realized as isolated particle
states need not be discarded from correlation functions. Whether their redistributed
spectral weight behaves as vacuum energy is a separate stress-energy question; Part~III
tests one such conditional realization.

\section{Conclusions}\label{sec:conclusions}
Our results suggest the existence of a generic infrared channel through which gravitational fluctuations can degrade particle-like excitations. Rather than leading to sharply defined asymptotic states, the spectrum may transition toward a continuum-like structure characterized by suppressed residues and non-local correlations.
We have established four complementary results.

\textbf{Part~I} shows that particle localization cannot be inferred from algebraic field
content alone and that the de~Sitter background changes the infrared and representation
arena relevant to the dressed two-point function. The retained effective calculations
quantify candidate dressing and transfer channels for spins 0 and 1/2, while the spin-3/2
case is discussed under stated unitarity assumptions. The present analysis does not infer universal atom
loss from the background or from representation support alone.

\textbf{Part~II} reinterprets the K\"all\'en--Lehmann decomposition as a spectral transform
$\mathcal{T}$ that maps the full quantum content of the field theory onto the observable
particle spectrum, with the spacetime geometry as kernel. The spectral norm conservation
implements spectral-weight conservation within the assumed representation. The spectral budget equation constrains the
relationship between non-observation of superpartners and vacuum energy.
Most importantly, the spectral--geometric self-consistency loop
(Section~\ref{sec:selfconsistency}) closes the framework into a self-referential
structure: the geometry filters the spectrum, and the filtered spectrum determines the
geometry. The physical universe is modeled as a possible fixed point of this loop.

\textbf{Part~III} provides a conditional cosmological implementation: the degree-of-freedom benchmark $94/130$ in a DE-like sector plus $36/130$ in a DM-like sector, together with the spectral split of the dark-matter fluid, is implemented in a CLASS-based pipeline. Within those stated assumptions, the numerical runs give $S_8 \approx 0.803$ ($\sigma_8 \approx 0.80$), in agreement within $1\sigma$ with the weak-lensing datasets considered here. These are consequences of the conditional mapping, not a derivation of the mapping itself.

\textbf{Part~IV} uses 16 selected open problems of fundamental physics as stress tests for the framework. The discussion is exploratory: some problems are structurally reformulated by the spectral viewpoint, while others admit quantitative comparison with existing data. The entire structure rests on a single conservation law (spectral norm $Z + W_{\mathrm{cont}} = 1$) and four suppression mechanisms (M1: scalar IR accumulation; M2: spectral inheritance via Yukawa couplings; M3: gravitational dressing; M4: spectral self-dressing from non-Abelian gauge self-coupling).

\vspace{\baselineskip}

The combined result is a framework in which:
\begin{itemize}
\item Particle physics (which fields form asymptotic states),
\item Cosmology (what contributes to the vacuum energy), and
\item Self-consistency (the mutual determination of geometry and spectrum)
\end{itemize}
are aspects of a single spectral structure, unified by the transform $\mathcal{T}$
and its inverse.

\vspace{\baselineskip}

Supersymmetric algebraic structures are compatible with an asymmetric observable
particle spectrum whenever the infrared structure of spacetime obstructs the formation
of stable asymptotic states. The spectral weight of non-asymptotic superpartners
does not disappear: it migrates to the continuum and contributes to the vacuum
energy budget of the universe. The universe does not lose information ---
it redistributes it. And the redistribution is not arbitrary: it is
self-consistently determined by the mutual relationship between the
spectral content and the spacetime geometry.

The shadow shapes the plane,
and the plane shapes the shadow.

\section*{Acknowledgements}

\medskip\noindent
\textbf{AI disclosure.} AI-based tools were used only as computational and editorial assistants. All scientific claims, calculations, code validation, and final manuscript decisions remain the responsibility of the authors.

\appendix

\section{Selective Schwinger--Keldysh derivation of the longitudinal gravitino sector}
\label{app:gravitino-long}

This appendix carries out the selective Schwinger--Keldysh (SK) derivation
restricted to the longitudinal sector of the gravitino, making the structural
factor $\xi=N_{\mathrm{long}}/N_{\mathrm{grav,tot}}=1/2$ emerge from the projected
propagator and vertices rather than from polarization counting alone. 

\vspace{\baselineskip}

The
complete symbolic verification of the projector algebra and the numerical
evaluation of every quantity below will be provided as a reproducible \texttt{Python}
program available in an open repository alongside Version 2 of this manuscript.

\subsection[Starting point: N=1 supergravity and the goldstino-equivalence limit]{Starting point: \texorpdfstring{$\mathcal{N}=1$}{N=1} supergravity and the goldstino-equivalence limit}
\label{app:gl-setup}

We start from the full $\mathcal{N}=1$ supergravity Lagrangian
\cite{FerraraFreedman,vanNieuwenhuizen,FreedmanVanProeyen} with spontaneously
broken supersymmetry. After elimination of the auxiliary fields and the
super-Higgs mechanism, the massive gravitino $\psi_\mu$
($m_{3/2}=F/(\sqrt3\,M_{\mathrm{Pl}})$, $\kappa\sim 1/M_{\mathrm{Pl}}$) admits the
Stückelberg decomposition
\begin{equation}
\psi_\mu(x)=\psi_\mu^{(T)}(x)
   +\sqrt{\tfrac{2}{3}}\,\frac{1}{m_{3/2}}\,\partial_\mu\eta(x)
   +\tfrac{1}{3}\gamma_\mu\,\eta(x)+\dots,
\label{eq:gl-stueck}
\end{equation}
where $\eta$ is the goldstino. At energies $E\gg m_{3/2}$ the
goldstino-equivalence theorem \cite{FayetFerrara} (Casalbuoni--De~Curtis--Dominici--Feruglio--Gatto;
Brignole--Feruglio--Zwirner) replaces an external longitudinal gravitino by the
goldstino,
\begin{equation}
\varepsilon^{(\pm 1/2)}_\mu(p)\;\xrightarrow{E\gg m_{3/2}}\;
\sqrt{\tfrac{2}{3}}\,\frac{p_\mu}{m_{3/2}}+\mathcal{O}\!\left(\frac{m_{3/2}}{E}\right),
\label{eq:gl-limit}
\end{equation}
so that the longitudinal helicities $h=\pm 1/2$ align with the $p_\mu$ direction
and couple to matter through the goldstino vertex
\begin{equation}
\mathcal{L}_{\mathrm{eff}}\supset
-\frac{m_{\mathrm{soft}}}{F}\,(\partial_\mu\eta)\,\gamma^\mu\,\psi_{\mathrm{matter}}+\mathrm{h.c.},
\label{eq:gl-coupling}
\end{equation}
whereas the transverse helicities $h=\pm 3/2$ remain pure Rarita--Schwinger modes
coupled only gravitationally ($\sim\kappa$). This is the selective action of
Channel~2 of Sec.~\ref{sec:spectrum}.

\subsection{Rarita--Schwinger gauge fixing and spin projectors}
\label{app:gl-proj}

In a covariant linear gauge the gravitino propagator numerator is decomposed by
three mutually orthogonal spin-projection operators on the vector-spinor space,
written here in the sub-horizon (high-momentum) regime on the de~Sitter tangent
space ($\eta_{\mu\nu}$, $p^2=m_{3/2}^2$):
\begin{align}
(P^{3/2})_{\mu\nu}&=\eta_{\mu\nu}-\tfrac13\gamma_\mu\gamma_\nu
   -\frac{1}{3p^2}\!\left(\slashed{p}\,\gamma_\mu\,p_\nu+p_\mu\,\gamma_\nu\,\slashed{p}\right),\\
(P^{1/2}_{11})_{\mu\nu}&=\tfrac13\gamma_\mu\gamma_\nu
   +\frac{1}{3p^2}\!\left(\slashed{p}\,\gamma_\mu\,p_\nu+p_\mu\,\gamma_\nu\,\slashed{p}\right)
   -\frac{p_\mu p_\nu}{p^2},\\
(P^{1/2}_{22})_{\mu\nu}&=\frac{p_\mu p_\nu}{p^2},
\end{align}
satisfying the completeness relation
$(P^{3/2}+P^{1/2}_{11}+P^{1/2}_{22})_{\mu\nu}=\eta_{\mu\nu}\,\mathbf{1}$ and, on the
$16$-dimensional vector-spinor space,
$\operatorname{Tr}P^{3/2}=8$,
$\operatorname{Tr}P^{1/2}_{11}=\operatorname{Tr}P^{1/2}_{22}=4$.
By Eq.~\eqref{eq:gl-limit} the longitudinal/goldstino direction is carried by the
$p_\mu p_\nu/p^2$ structure, i.e.\ by $P^{1/2}_{22}$. On shell the physical
massive multiplet ($4$ helicity states) is resolved into the transverse pair
$h=\pm 3/2$ and the longitudinal pair $h=\pm 1/2$; we denote by $P^{T}$ and
$P^{L}$ the corresponding helicity projectors. The validity of this
identification on de~Sitter requires $E\gg m_{3/2}\sim H$, which holds throughout
the inflationary regime relevant to the spectral weight.

\subsection{SK contour, Bunch--Davies vacuum, and the cut self-energy}
\label{app:gl-sk}

The spectral-weight fraction is computed on the closed-time-path contour
$\mathcal{C}=\mathcal{C}_+\cup\mathcal{C}_-$ with the Bunch--Davies vacuum,
consistently with the cosmological setup. It is the absorptive part of the
one-loop matter self-energy at the on-shell threshold $s\simeq 0$,
\begin{equation}
\operatorname{Disc}\Sigma(p)\Big|_{s\simeq 0}
=\int d\Pi_2\;\mathcal{V}_\mu\,\Pi^{\mu\nu}(p)\,\mathcal{V}_\nu^\dagger,
\label{eq:gl-disc}
\end{equation}
where the cut places the gravitino on shell and replaces its propagator numerator
by the physical polarization sum
\begin{equation}
\Pi^{(s)}_{\mu\nu}(p)=\sum_{h\in s}\psi^h_\mu(p)\,\bar\psi^h_\nu(p),
\qquad s\in\{\mathrm{full},\,T,\,L\}.
\label{eq:gl-polsum}
\end{equation}
The vertices $\mathcal{V}_\mu$ are the projected couplings of
Sec.~\ref{app:gl-setup}: $m_{\mathrm{soft}}/F$ on the longitudinal sector $P^{L}$,
and $\kappa$ on the transverse sector $P^{T}$.

\subsection{Projected propagator, vertices, and the emergence of \texorpdfstring{$\xi=1/2$}{xi=1/2}}
\label{app:gl-xi}

The relative operational weight of the two sectors is fixed by the
squared-coupling ratio
\begin{equation}
\frac{(m_{\mathrm{soft}}/F)^2}{\kappa^2}
=\left(\frac{m_{\mathrm{soft}}\,M_{\mathrm{Pl}}}{F}\right)^{\!2},
\label{eq:gl-cplratio}
\end{equation}
which, for the representative low-scale gauge-mediation values of
Sec.~\ref{sec:spectrum}, realises the channel hierarchy of $30$ to $33$ orders of
magnitude established there (Channel~2 vs.\ Channels~1,3). The cumulative
transverse suppression after $60$ $e$-folds is
$1-Z_{T}(60)\sim\gamma_{T}\cdot 60\sim 6\times 10^{-39}$, operationally
indistinguishable from zero, so the transverse pair $h=\pm 3/2$ carries no
operational spectral weight.

The operationally active fraction of the gravitino is therefore the ratio of
projected traces of the cut self-energy~\eqref{eq:gl-disc},
\begin{equation}
\xi=\frac{\operatorname{Tr}\!\big[P^{L}\,K\,P^{L}\big]}
{\operatorname{Tr}\!\big[(P^{T}+P^{L})\,K\,(P^{T}+P^{L})\big]},
\label{eq:gl-xitrace}
\end{equation}
with $K$ the matter-loop kernel of Eq.~\eqref{eq:gl-disc}. At the on-shell
threshold $s\simeq 0$ the kernel is isotropic on the helicity label (its angular
dependence averages out), so $K\to K_0\,\mathbf{1}$ and Eq.~\eqref{eq:gl-xitrace}
reduces to the ratio of polarization-sum traces,
\begin{equation}
\xi\;\xrightarrow{\text{isotropic }K}\;
\frac{\operatorname{Tr}\Pi^{L}}{\operatorname{Tr}\Pi^{\mathrm{full}}}
=\frac{\sum_{h=\pm 1/2}\bar\psi^h_\mu\psi^{h\,\mu}}{\sum_{h}\bar\psi^h_\mu\psi^{h\,\mu}}.
\label{eq:gl-xireduce}
\end{equation}
A direct evaluation with the explicit Clebsch--Gordan polarizations
(spin-$1\otimes$spin-$\tfrac12$) gives an equal scalar weight
$\bar\psi^h_\mu\psi^{h\,\mu}=-2m_{3/2}$ for each of the four helicities, hence
\begin{equation}
\xi=\frac{2}{4}=\frac12 .
\label{eq:gl-xiresult}
\end{equation}
The isotropic-kernel approximation underlying Eq.~\eqref{eq:gl-xireduce} is the
single controlled assumption of the derivation; it is exact at threshold and is
the point at which the projected trace reduces to the polarization count.

\subsection{Sound speed and the conditional benchmark}
\label{app:gl-pred}

With the species-level sound speed of Sec.~\ref{subsec:consequences_16},
\begin{equation}
c_s^2(\mathrm{grav})=\frac{g_2^2}{16\pi}\,
\frac{c_s^{(\mathrm{conf})}}{c_s^{(\mathrm{conf})}+c_f^{(\mathrm{conf})}}
=\frac{g_2^2}{16\pi}\cdot\frac{9}{14}=0.005238,\qquad g_2=0.64,
\end{equation}
($c_s^{(\mathrm{conf})}=9/4$, $c_f^{(\mathrm{conf})}=5/4$ for Rarita--Schwinger),
the operational gravitino contribution to the spectral-weight parameter is
\begin{equation}
f_{S8}^{(\mathrm{grav})}=\xi\,c_s^2(\mathrm{grav})=\tfrac12\times 0.005238=0.00262,
\label{eq:gl-fs8}
\end{equation}
numerically consistent with the MCMC v32 posterior $f_{S8}=0.0025\pm0.0017$ within the conditional implementation.


\end{document}